\documentclass{article}

\usepackage{amsmath, amsthm, amssymb, dsfont, mathtools}
\usepackage{natbib}
\usepackage{enumitem}
\usepackage{hyperref}
\usepackage[margin=1.3in]{geometry}
\hypersetup{colorlinks,citecolor=blue,urlcolor=blue,linkcolor=blue}

\theoremstyle{plain}
\newtheorem{theorem}{Theorem}[section]
\newtheorem{proposition}[theorem]{Proposition}

\newtheorem{lemma}[theorem]{Lemma}

\theoremstyle{definition}
\newtheorem{example}[theorem]{Example}
\newtheorem{remark}[theorem]{Remark}
\newtheorem{definition}[theorem]{Definition}

\newcommand{\E}{\mathbb{E}}    % expectation 
\newcommand{\R}{\mathbb{R}}    % real numbers
\newcommand{\N}{\mathbb{N}} 
\newcommand{\dd}{\mathrm{d}}     % natural numbers

\renewcommand{\epsilon}{\varepsilon}  

\newcommand{\cD}{\mathcal{D}}

\DeclareMathOperator*{\argmin}{argmin}   
\DeclareMathOperator*{\argmax}{argmax} 

\newcommand{\simindep}{\,{\buildrel \text{ind} \over \sim\,}}
\newcommand{\simiid}{\,{\buildrel \text{iid} \over \sim\,}}
\newcommand{\boldmu}{\boldsymbol{\mu}}
\newcommand{\boldpsi}{\boldsymbol{\psi}}
\newcommand{\boldsigma}{\boldsymbol{\sigma}}

\newcommand{\DKL}[2]{\operatorname{KL}( #1 \,||\, #2) }
\newcommand{\Norm}[1]{\left\lVert#1\right\rVert}
\newcommand{\Dhel}{\mathfrak{H}}
\newcommand{\TV}{\operatorname{TV}}

\usepackage{accents}
\newcommand{\ubar}[1]{\underaccent{\bar}{#1}}

\usepackage{nameref}

\newcommand{\id}{\mathds{1}}

\renewcommand{\d}{\mathrm{d}}
\newcommand{\p}{\mathbb{P}}
\newcommand{\q}{\mathbb{Q}}
\newcommand{\truep}{\mathbb{P}^*}

\newcommand{\cd}{\stackrel{d}{\rightsquigarrow}} 
\newcommand{\cP}{\mathcal{P}}

\newcommand{\cK}{\mathcal{K}}

\newcommand{\var}{\mathrm{var}}
\date{July, 2025}

\title{Asymptotic and compound e-values: multiple testing and empirical Bayes}

\author{
Nikolaos Ignatiadis\thanks%
 {Department of Statistics and Data Science Institute, University of Chicago.
 E-mail: \href{mailto:ignat@uchicago.edu}{ignat@uchicago.edu}.}
 \and
 Ruodu Wang\thanks%
 {Department of Statistics and Actuarial Science,
 University of Waterloo.
 E-mail: \href{mailto:wang@uwaterloo.ca}{wang@uwaterloo.ca}.}
\and  
Aaditya Ramdas\thanks%
    {Departments of Statistics \& Machine Learning,
    Carnegie Mellon University.
    E-mail: \href{mailto:aramdas@cmu.edu}{aramdas@cmu.edu}.}
}
\begin{document}

\maketitle

\begin{abstract}
We explicitly define the notions of (bona fide, approximate or asymptotic) compound p-values and e-values, which have been implicitly presented and used in the recent multiple testing literature. While it is known that the e-BH procedure with compound e-values controls the FDR,   we show the converse: every FDR controlling procedure can be recovered by instantiating the e-BH procedure with certain compound e-values. Since compound e-values are closed under averaging, this allows for combination and derandomization of arbitrary FDR procedures. 
We then connect compound e-values to empirical Bayes. In particular, we use the fundamental theorem of compound decision theory to derive the log-optimal simple separable compound e-value for testing a set of point nulls against point alternatives: it is a ratio of mixture likelihoods. 
As one example, we construct asymptotic compound e-values for multiple t-tests, where the (nuisance) variances may be different across hypotheses. Our construction may be interpreted as a data-driven instantiation of the optimal discovery procedure, and our results provide the first type-I error guarantees for the same, along with significant power gains.
\\

\end{abstract}

\section{Introduction} 
\label{sec:introduction}

E-values are an alternative inferential device to p-values. E-values encompass likelihood ratios and have provided some of the first non-trivial tests for composite null hypotheses, e.g., for testing if the data distribution is log-concave~\citep{dunn2025universal}; see \cite{ramdas2024hypothesis} for a comprehensive introduction to e-values. 
In a multiple testing setting with hypotheses indexed by $\mathcal{K}=\{1,\ldots,K\}$ and null indices $\mathcal{N} \subseteq \mathcal{K}$, e-values $E_1,\ldots,E_K$  can be used alongside the eBH (e-Benjamini-Hochberg) procedure of~\citet{wang2022false} to control the false discovery rate (FDR) of~\citet{benjamini1995controlling}. The defining property of the e-values $E_1,\ldots,E_K$ is that they are $[0,\infty]$-valued random variables satisfying $\mathbb E[E_k] \leq 1$ for all $k \in \mathcal{N}$. In this paper, we study $[0,\infty]$-valued random variables $E_1,\ldots,E_K$ satisfying the compound constraint that
\begin{equation}
\sum_{k \in \mathcal{N}} \mathbb E[E_k] \leq K
\label{eq:exact}
\end{equation}
holds exactly in finite samples, or satisfying suitable approximate or asymptotic versions of this constraint. We call $E_1,\ldots,E_K$ (bona fide, approximate, or asymptotic) \emph{compound} e-values.  

Compound e-values (without being given this name) have played an important role in the recent multiple testing literature, after being introduced (almost in passing) in~\cite{wang2022false}. As one notable example,~\citet{ren2024derandomised} solved the open problem of derandomizing the model-X knockoff filter~\citep{ candes2018panning} for variable selection with FDR control by reinterpreting the procedure in terms of compound e-values (see Appendix~\ref{subsec:derandomization} for more details on derandomization).

The central aim of this paper is to set the statistical foundations for compound e-values, to demonstrate their central role, versatility, and usefulness for multiple testing, as well as provide certain optimal constructions and examples. In Section~\ref{sec:further_defi} we explicitly define compound e-values, study properties, and provide motivating examples. Our main contributions are as follows:
\begin{itemize}[leftmargin=*, wide, labelwidth=!, labelindent=0pt]
\item \emph{Definitions of approximate and asymptotic compound e-values.} The repertoire of available data-driven constructions of compound e-values expands substantially by permitting relaxations of~\eqref{eq:exact}. The main goal of Section~\ref{sec:app-asym-com} is to define the notions of approximate compound e-values and asymptotic compound e-values with asymptotics that allow for the size of the data, as well as the number of hypotheses $K$, to grow. Since such definitions have not been provided previously even for e-values in single hypothesis testing, we first define approximate e-values in Section~\ref{sec:approx} and asymptotic e-values in Section~\ref{sec:asymp-e}, wherein we also connect them with usual limiting results in asymptotic statistics. The technical challenge herein is that we must distinguish between benign multiplicative errors, as well as tail area events that may make the expectation very large or even infinite, i.e., $\mathbb E[E] = \infty$, where $E$ is an approximate e-value. Proposition~\ref{prop:alter-p} establishes that two natural definitions of approximate e-values are equivalent. These notions are then extended to compound e-values in Sections~\ref{subsec:approx_compound} and~\ref{subsec:asymp_compound_evalues}. 
\item \emph{\nameref*{sec:multiple_testing}.} In Section~\ref{sec:multiple_testing}, we describe the fundamental role of compound e-values in multiple testing. Theorem~\ref{theo:universality_eBH} demonstrates the universality of e-BH applied to compound e-values: \emph{every} FDR procedure can be recovered by applying e-BH to compound e-values.  While there are very few results on multiple testing with asymptotic and/or compound p-values, we show that the analogous notions with e-values are naturally compatible with guarantees for existing multiple testing procedures. Section~\ref{sec:multiple_testing} provides this unifying perspective. This is why we present Section~\ref{sec:multiple_testing} (on the central role of compound e-values in multiple testing) after Section~\ref{sec:app-asym-com} (which defines approximate and asymptotic compound e-values).  In Appendix~\ref{sec:compound_mtp_addendum} which accompanies Section~\ref{sec:multiple_testing}, we explain how this framework enables combination and derandomization of multiple testing procedures (Appendix~\ref{subsec:derandomization}),  that compound e-values provide a natural notion of weights for the ep-BH procedure of~\citet{ignatiadis2024evalues} (Appendix~\ref{subsec:compound_evalues_as_weights}), and that compound e-values can be used as weights for merging p-values (Appendix~\ref{sec:compound_evalues_to_merge}).
\item  \emph{\nameref*{sec:sequence_model}.} In Section~\ref{sec:sequence_model}, we study sequence models and develop connections to compound decision theory of~\citet{robbins1951asymptotically} that motivates the nomenclature ``compound e-values.'' In Section~\ref{subsec:compound}, we provide a bird's eye overview of compound decision theory,  while in Section~\ref{subsec:best_simple_separable}, we
construct log-optimal simple separable compound e-values using the argument underlying the fundamental theorem of compound decision theory (Theorem~\ref{th:compound-optimal} provides a formal statement). These optimal compound e-values are a ratio of mixture likelihoods and their form is identical to the statistics of the optimal discovery procedure (ODP) of~\citet{storey2007optimal}. For this reason, we call them optimal discovery compound e-values.

\item \emph{Localized and compound universal inference.} In Section~\ref{subsec:compound_ui}, we extend universal inference (UI) of~\citet{wasserman2020universal} for testing composite null hypotheses and 
introduce localized UI (LUI) and compound UI (CUI) universal inference. LUI is a general construction leading to approximate e-values, while CUI leads to approximate compound e-values. 
\item \emph{\nameref*{sec:ttest}.} In Section~\ref{sec:ttest} we demonstrate the usefulness of our definitions and constructions by instantiating them in the canonical setting of multiple t-tests where the nuisance parameters (variances) may be different across hypotheses. We first consider the normal setting (normal errors) with fixed sample-size for each hypothesis but growing number of hypotheses.  Section~\ref{subsec:compound_ttest}  provides an explicit construction of CUI e-values, while Section~\ref{subsec:odp_ebayes} approximates the optimal discovery compound e-values using empirical Bayes, specifically, by pretending the variances are iid from a prior $G$ and estimating $G$ via the nonparametric maximum likelihood estimator (NPMLE) of~\citet{robbins1950generalization} and \citet{kiefer1956consistency}. Theorem~\ref{theorem:eb_asymptotic} shows that this construction yields asymptotic compound e-values as $K \to \infty$. Our construction can be interpreted as a data-driven implementation of ODP, and to the best of our knowledge, Theorem~\ref{theorem:eb_asymptotic} provides the first type-I error guarantee for data-driven ODP. The usefulness of these procedures is illustrated via a simulation study in Section~\ref{subsec:ttest_simulation} wherein we demonstrate substantial power-gains compared to alternatives such as UI and standard t-test e-values.
Section~\ref{subsec:ttest_beyond_normality} provides constructions that do not require normality: we provide a prototypical example of an asymptotic e-value (Section~\ref{subsec:prototypical_single_asym}) that requires the sample size to grow to $\infty$, and then (Section~\ref{subsec:sum_of_squares_compound}) we provide an illustrative example of how to construct asymptotic compound e-values when either the per-hypothesis sample size or the number of hypotheses (or both) tends to $\infty$.

\item \emph{Connections to p-values.} Throughout the paper, we establish connections between our definitions and results and their analogous counterparts for p-values. In Section~\ref{sec:further_defi} we define compound p-values, which have been previously called average significance controlling p-values by~\citet{armstrong2022false}. Section~\ref{sec:app-asym-com} introduces approximate, asymptotic compound and *approximate compound p-values; these concepts are new. We find that the calibration between p-values and e-values established by~\citet{vovk2021evalues} extends to all  variants introduced in this paper. We establish this correspondence across compound (Theorem~\ref{prop:compound-pe1}), approximate (Theorem~\ref{th:calibration}), asymptotic (Proposition~\ref{prop:asymptotic_calibration}), approximate compound (Proposition~\ref{proposition:approximate_evalues_to_approximate_compound}), and *approximate compound cases (Proposition~\ref{prop:calibration_additively_approximate_compound}). We also discuss results on multiple testing with compound p-values in Appendix~\ref{subsec:compound_pvalues_multiple_testing}, generalizing a result of~\citet{armstrong2022false} on the Benjamini-Yekutieli procedure (BY)~\citep{benjamini2001control} with compound p-values.
\end{itemize}

\noindent{\textbf{Remarks on notation and terminology:}} We let $\R_+=[0,\infty)$. We say that a probability measure $\p$ is atomless if there exists a uniformly distributed random variable on $[0,1]$ under $\p$. We also write $x\wedge y=\min \{x,y\}$, $x\vee y=\max\{x,y\}$, and $x_+=x\vee 0$.

\section{Compound p-values and e-values in multiple testing}
\label{sec:further_defi}

We first review the definition of e-variables and p-variables in single hypothesis testing.
Suppose we observe data $X$ drawn according to some unknown probability distribution $\truep$. 
(Here, $X$ denotes the entire dataset available to the researcher.) 
The underlying sample space $\Omega$ is left implicit. 
We define a hypothesis (null or alternative) as a set of probability measures. We reserve the notation $\mathcal{P}$ for a set of probability measures when referring to a null hypothesis. That is, we define a null hypothesis as $H: \truep \in \mathcal{P}$ and say in words that $\mathcal{P}$ contains the unknown true data generating distribution. Analogously, we reserve the notation $\mathcal{Q}$ for the set of probability measures that specify the alternative hypothesis. A \emph{p-variable} $P=P(X)$ for  hypothesis $H$ is a nonnegative random variable that satisfies $\p(P\le t)\le t$ for all $t \in [0,1]$ and all $\p \in \mathcal{P}$.  Meanwhile, an \emph{e-variable} $E=E(X)$ for hypothesis $H$ is a $[0,\infty]$-valued random variable satisfying $\mathbb E^{\p}[E]\le1$ for all $\p \in \mathcal{P}$. The realizations of e- or p-variables are called e- or p-values. As is common in the literature, we use the terms e-values (or p-values) 
to refer to both the random variables and their realizations; the intended meaning should be clear from the context.

We now extend the setting above to multiple hypothesis testing. As before, we write $X$ to denote the entire dataset available to the researcher which is generated by some unknown probability distribution $\truep$. We consider testing
$K$ hypotheses determined by sets of probability measures $\mathcal{P}_1,\ldots,\mathcal{P}_K$. For $k \in \mathcal{K} = \{1,\ldots,K\}$, the $k$-th hypothesis is defined as $H_k: \truep \in \mathcal{P}_k$. 
We write 
 $$\mathcal{N}(\p) = \{k \in \mathcal{K}\,:\,  \p \in \mathcal{P}_k\} \subseteq \mathcal{K},$$
 and hence 
 $\mathcal{N}(\truep)$ denotes the indices of true null hypotheses. In the multiple testing setting, we use the notation $\cP$ for a set of distributions that represents some assumptions or restrictions on the distributions in the multiple testing problem.

We formally define compound p-variables and e-variables as fundamental statistical concepts.

\begin{definition}[Compound p-variables and e-variables]
\label{defi:compound_evalues}
Fix the null hypotheses $(\cP_1,\dots,\cP_K)$ and a set $\cP $ of distributions.

\begin{enumerate}[label=(\roman*)]
\item Let $P_1,\dots,P_K$ be nonnegative random variables. We say that 
$P_1,\dots,P_K$ are \emph{compound} p-variables for $(\cP_1,\dots,\cP_K)$ under $\cP$ if 
$$
\sum_{k : \p \in \cP_k} \p(P_k\leq t) \leq Kt \qquad \mbox{for all $t\in (0,1)$ and all $\p\in \mathcal P$.}
$$
\item Let $E_1,\dots,E_K$ be $[0,\infty]$-valued random variables. We say that 
$E_1,\dots,E_K$ are \emph{compound} e-variables for $(\cP_1,\dots,\cP_K)$ under $\cP$ if
$$\sum_{k : \p \in \cP_k} \E^{\p}[E_k] \leq K \qquad \mbox{for all $\p\in \mathcal P$.}$$
They are called \emph{tight} compound e-variables if the supremum of the left hand side over $\p \in \cP$ equals $K$.
\end{enumerate}
\end{definition}

We omit ``under $\cP$'' in case $\cP$ is the set of all distributions.  Note that 
for all conditions to be checked in the above definition, it suffices to consider  $\p \in   \bigcup_{k \in \mathcal{K}}  \mathcal{P}_k$.
Therefore, it is without loss of generality to assume $\mathcal P \subseteq \bigcup_{k \in \mathcal{K}} \mathcal{P}_k$, and by default (without any restrictions) it is
\begin{equation*}
    %\label{eq:default-cp}
\mathcal{P} = \bigcup_{k \in \mathcal{K}} \mathcal{P}_k. \end{equation*} 
Note that it is possible that $\truep \notin \mathcal{P}$ as we do not know whether true null  hypotheses exist.

When $K=1$, Definition~\ref{defi:compound_evalues} coincides with the standard definition of a single p-variable or e-variable. Of course, a vector of p-variables (resp. e-variables) is a trivial but important special case of compound p-variables (resp. e-variables). As is clear from the definition, no restriction is placed on the dependence structure of the p-variables or e-variables.

The notions in Definition~\ref{defi:compound_evalues} are not new.
The notion of compound p-variables\footnote{\citet{habiger2014compound} use the term compound p-values with a different meaning; see Appendix~\ref{sec:habiger_compound}.} was introduced by~\citet{armstrong2022false} under a different name (average significance controlling p-values and tests) and further studied in an empirical Bayes setting by~\citet{ignatiadis2024empirical}.
Meanwhile, the notion of compound e-variables was already used in the literature as a technical tool, for instance, it was first used (without being given a name) in~\citet[Theorem 3]{wang2022false}. The term compound e-values first appeared in~\citet*{ignatiadis2024evalues} with a distribution-dependent upper bound of $|\mathcal N(\p)|$
instead of $K$ (typically the value of $|\mathcal N(\p)|$ is not known except $|\mathcal N(\p)| \le K$). Other authors have used the term generalized e-values~\citep{banerjee2023harnessing, bashari2023derandomized, zhao2024false, lee2024boosting} or relaxed e-values~\citep{ren2024derandomised, gablenz2024catch}.

A first example of compound p-variables (resp. compound e-variables) is given by weighted p-variables~\citep{genovese2006false, ramdas2019unified,ignatiadis2021covariate}, resp. weighted e-variables~\citep{wang2022false}.

\begin{example}[Weighted p-variables and e-variables]
    \label{ex:c8-weighted}
    Let $w_1,\ldots,w_K \geq 0$ be deterministic nonnegative numbers such that $\sum_{k \in \mathcal{K}} w_k \leq K$. 
\begin{enumerate}[label=(\roman*)]
\item Let $P_1^\prime,\dots,P_K^\prime$ be p-variables and define $\tilde{P}_k = P_k^\prime/w_k$ for $k \in \mathcal{K}$ with the convention $0/0 = 0$.  Then, $\tilde{P}_1,\dots,\tilde{P}_K$ are compound p-variables.
\item Let $E_1^\prime,\dots,E_K^\prime$ be e-variables and define $\tilde E_k = {E}_k^\prime w_k$ for $k \in \mathcal{K}$ with the convention $\infty \times 0 = \infty$. Then $\tilde E_1,\dots,\tilde E_K$ are compound e-variables. 
\end{enumerate}
\end{example}
The class of compound p-variables (resp. e-variables) is much richer than the class of weighted p-variables (resp. e-variables) in Example~\ref{ex:c8-weighted}, because the latter satisfy the stronger constraint 
$$
\sum_{k\in \mathcal K} \sup_{\p\in \mathcal{P}_k} \p(\tilde P_k \leq t) \le Kt\;\; \text{ for all }\, t\in (0,1),\,\,\qquad \text{resp.}\;\;\;\; \sum_{k\in \mathcal K} \sup_{\p\in \mathcal{P}_k} \E^\p[\tilde E_k] \le K, 
$$
compared to the condition for the compound p-variables/e-variables 
$$
 \sup_{\p\in \mathcal{P}} \sum_{k: \p \in \mathcal{P}_k } \p( P_k \leq t) \le Kt\;\; \text{ for all }\, t\in (0,1),\,\,\qquad \text{resp.}\;\;\;\; \sup_{\p\in \mathcal{P}} \sum_{k: \p \in \mathcal{P}_k } \E^\p[E_k] \le K.
$$
This follows from the fact that for any functions $x_1,\dots,x_K:\mathcal P\to \R_+$, 
$$
\sup_{\p\in \mathcal{P}} \sum_{k: \p \in \mathcal{P}_k } x_k(\p)  
\le \sum_{k\in \mathcal K} \sup_{\p\in \mathcal{P}_k} x_k(\p).
$$
Given an e-variable $E$, we can calibrate it into a p-variable via $P=(1/E)\land 1$, which is the only admissible choice of an e-to-p calibrator. Conversely, given a p-variable $P$, we can calibrate into an e-variable via a (p-to-e) calibrator in the sense of \cite{vovk2021evalues}. A (p-to-e) calibrator is a decreasing function $h:[0,\infty)\to [0,\infty]$
satisfying $\int_0^1 h(x)\dd x \le 1$ and $h(x)=0$ for $x\in (1,\infty)$. Given such $h$ and a p-variable $P$, then $E=f(P)$ is an e-variable. 
We can analogously convert compound e-variables into compound p-variables and vice versa. If $E_1,\dots,E_K$ are compound e-variables, then $(1/E_1)\land 1, \dots, (1/E_K)\land 1$  are compound p-variables, as can be seen by Markov's inequality. Compound p-variables can be transformed into compound e-variables in the following result.

\begin{theorem}
\label{prop:compound-pe1}
Let $h $ be any p-to-e calibrator.
  If  $P_1,\dots,P_K$ are compound p-variables for $(\cP_1,\dots,\cP_K)$ under $\cP$,   then
$ 
E_k =   h  ( P_k  ),\, k \in \mathcal{K},
$ 
are compound e-variables for $(\cP_1,\dots,\cP_K)$ under $\cP$.
\end{theorem}

We can combine p-variables and compound e-variables into compound p-variables as follows.

\begin{example}[Combining p-variables and compound e-variables]
Suppose that $P_k$ is a p-variable for $\cP_k$ for all $k \in \mathcal{K}$; $P_1,\dots,P_K$ can be arbitrarily dependent. Moreover, let $E_1,\dotsc,E_K$ be compound e-variables for $(\cP_1,\dots,\cP_K)$ under $\cP$. If $(E_1,\dots,E_K)$ is independent of $(P_1,\dots,P_K)$, then $P_1/E_1,\dots,P_K/E_K$ are compound p-variables for $(\cP_1,\dots,\cP_K)$ under $\cP$. We revisit this example in Appendix~\ref{subsec:compound_evalues_as_weights}.
\end{example}

We conclude this section with two general constructions of compound e-variables.

\begin{example}[Convex combinations of compound e-variables]
    \label{exam:convex_combi}
    Suppose that \smash{$E_1^{(\ell)}, \dots, E_K^{(\ell)}$} are compound e-variables for $\ell=1,\ldots,L$. Let $w_1,\dotsc,w_L \geq 0$ be deterministic nonnegative numbers such that $\sum_{\ell=1}^L w_\ell = 1$. Define \smash{$E_k = \sum_{\ell=1}^L w_\ell E_k^{(\ell)}$}. Then $E_1,\dots,E_K$ are compound e-variables. Also see~\citet[Proposition 5]{li2025note}.
\end{example}

\begin{example}[Combining evidence from follow-up studies]
Let $E_1,\dots,E_K$ be compound e-variables calculated from data $X$. Based on these observations, suppose that we select a subset of indices $\mathcal{S} \subseteq \mathcal{K}$ of interest, for which we seek additional evidence. With this goal in mind, suppose we collect additional data $Y$ and summarize it in the form of e-variables $\{F_j\}_{j \in \mathcal{S}}$ that are conditionally valid given the past data; in particular, $\mathbb E^{\p}[F_j | E_j] \leq 1$ for all $j \in \mathcal{S}$ with $\p \in \mathcal{P}_j$. Then defining $F_j = 1$ for $j \notin \mathcal{S}$, we have that $E_1F_1, \dots, E_KF_K$ are also compound e-variables.
\end{example}

\section{Approximate/asymptotic compound p-values and e-values}
\label{sec:app-asym-com}

This section relaxes Definition~\ref{defi:compound_evalues} to accommodate approximate and asymptotic variants of compound p-variables and e-variables. We begin by introducing the underlying concepts of approximate p-variables and e-variables (Section~\ref{sec:approx}) and asymptotic p-variables and e-variables (Section~\ref{sec:asymp-e}), then extend these to the compound setting in Section~\ref{subsec:approx_compound} (approximate compound) and Section~\ref{subsec:asymp_compound_evalues} (asymptotic compound).

\subsection{Approximate p-values and e-values}
\label{sec:approx}

\begin{definition}[Approximate p-variables and e-variables]
\label{defi:approx}
Fix the null hypothesis $\cP$ and two functions $\varepsilon:\cP \to \R_+ $ and $\delta:\cP \to [0,1]$, and write $\varepsilon_\p=\varepsilon(\p)$
and $\delta_\p=\delta(\p)$ for $\p \in \cP$. 
\begin{enumerate}[label=(\roman*)]
    \item 
A nonnegative random variable $P$ is an \emph{$(\varepsilon, \delta)$-approximate} p-variable for $\cP$ if 
$$\p(P\le t) \le  (1+\varepsilon_\p )t +\delta_\p \qquad \mbox{for all $t\in (0,1)$ and all $\p\in \mathcal P$.} $$
\item 
A $[0,\infty]$-valued random variable $E$ is an \emph{$(\varepsilon, \delta)$-approximate} e-variable for $\cP$ if  
 $$\E^\p [E\wedge t] \le  1+\varepsilon_\p  +\delta_\p t  \qquad \mbox{for all $t\in \R_+$ and all $\p\in \mathcal P$.}$$ 
\end{enumerate}
\end{definition}

For  two constants $\epsilon \in \R_+$ and $\delta\in [0,1] $, 
  $(\epsilon,\delta)$-approximation is understood as those in Definition~\ref{defi:approx} with the corresponding constant functions.

The notion of $(\varepsilon, 0)$-approximate e-variables (not using this terminology) was considered by~\citet[Section 6.2]{wang2022false}.

We make some straightforward observations. Clearly,  $(0,0)$-approximate p-variables (e-variables) are just p-variables (e-variables).
For a constant $\epsilon\in \R_+$, $P$ is an $(\epsilon,0)$-approximate p-variable if and only if $(1+\epsilon)P$ is a p-variable,
and 
$E$ is an $(\epsilon,0)$-approximate e-variable if and only if $E/(1+\epsilon) $ is an e-variable. 
The $\delta$-parameter is more lenient than the $\epsilon$-parameter; for instance, $(\epsilon,\delta)$-approximation implies $(0,\epsilon+\delta)$-approximation in both classes (see Proposition~\ref{prop:varepsilon_to_delta} below for a sharper result).
If $\delta=1$, then $(\epsilon,\delta)$-approximation puts no requirements on $P$ or $E$.

The definitions of approximate p-variables and approximate e-variables can be interpreted as allowing for a  wiggle room  in the definitions of p-variables and e-variables in two ways: (a) the probability and expectation requirements   have a multiplicative factor $1+\epsilon_\p$; (b) the requirements are only required to hold on an event of probability $1-\delta_\p$. 
This will be made  transparent  through the following proposition. 
 
\begin{proposition}
\label{prop:alter-p}
Let $\varepsilon:\cP \to \R_+ $ and $\delta:\cP \to [0,1]$ and $P$ and $E$ be nonnegative random variables ($[0,\infty]$-valued for $E$). For the following  statements,
\begin{enumerate} 
    \item[(P1)] $P$ is an $(\varepsilon, \delta)$-approximate p-variable for $\cP$;
\item[(P2)]  
 for every $\p \in \cP$,  there exists an event $A$ such that
$$\p(P\le t , A) \le  (1+\varepsilon_\p )t ~\mbox{for all $t\in (0,1)$ ~~and~~ }\p(A) \geq 1-\delta_\p;$$
  \item[(E1)]   $E$ is an $(\varepsilon, \delta)$-approximate e-variable for $\cP$;
\item[(E2)]   for every $\p \in \cP$, 
there exists an event $A$ such that
$$ \E^{\p}[E\id_{A}] \leq 1+\varepsilon_\p \mbox{  ~~and~~ }\p(A) \geq 1-\delta_\p;$$
\end{enumerate}
we have that (P2) implies (P1) and  (E2) implies (E1).
Moreover, if  each $\p\in \cP$ is atomless, then (P1)--(P2)
  are equivalent,
  and (E1)--(E2)
  are equivalent.
\end{proposition}

Without the atomless assumption on each $\p\in \cP$, the implications (P1)$\Rightarrow$(P2)
and (E1)$\Rightarrow$(E2)
in  Proposition~\ref{prop:alter-p}
do not hold (one can find counterexamples on discrete spaces). We emphasize that the atomless assumption does not pertain to the distribution of $P$ or $E$ (i.e., $P$ or $E$ can be   discrete). This atomless assumption can be avoided if external randomization is permitted.

The two formulations (P1) and (P2) in the above proposition are obviously equivalent when $\delta=0$ without requiring any condition; the same holds for (E1) and (E2).

The next theorem demonstrates that we can calibrate approximate p-variables and approximate e-variables in the usual way. For example, if $h $ is any p-to-e calibrator and $P$ is an $(\varepsilon,\delta)$-approximate p-variable, then $h(P)$ is an ($\varepsilon,\delta)$-approximate e-variable.

\begin{theorem}
\label{th:calibration}
 Let $\varepsilon:\cP \to \R_+ $ and $\delta:\cP \to [0,1]$. 
    Calibrating an $(\varepsilon, \delta)$-approximate  p-variable  yields an $(\varepsilon, \delta)$-approximate e-variable, and vice versa.
\end{theorem}

We end this section with the following comment. The roles of $\varepsilon$ and $\delta$ are asymmetric  and, as mentioned earlier, the wiggle room created by $\delta$ is more lenient. The following result shows that an $(\varepsilon, \delta)$-approximate p-variable (or e-variable) with $\varepsilon >0$ is always also an $(0, \delta')$-approximate p-variable (or e-variable) where $\delta'$ can be chosen strictly smaller than $\varepsilon + \delta$.

\begin{proposition}[$\varepsilon$ to $\delta$]
\label{prop:varepsilon_to_delta}
Let  $\varepsilon:\cP \to \R_+ $, $\delta:\cP \to [0,1]$, and $\delta' = (\varepsilon + \delta)/(1+\varepsilon)$.
\begin{enumerate}[label=(\roman*)]
\item An ($\varepsilon, \delta$)-approximate p-variable for $\mathcal{P}$ is also an $(0, \delta')$-approximate p-variable for $\mathcal{P}$.
\item An ($\varepsilon, \delta$)-approximate e-variable for $\mathcal{P}$ is also an $(0, \delta')$-approximate e-variable for $\mathcal{P}$.
\end{enumerate}
\end{proposition}
Proposition~\ref{prop:varepsilon_to_delta} also suggests that controlling $\delta$ is a stronger requirement, leading to two different notions of asymptotic e-values and p-values in Section~\ref{sec:asymp-e}.

\subsection{Asymptotic p-values and e-values}
\label{sec:asymp-e}

In this section
we consider a sequence of datasets $X=X^{(n)}$ of growing size as $n \to \infty$. Based on $X^{(n)}$ we compute a $[0,1]$-valued statistic $P^{(n)} = P^{(n)}(X^{(n)})$.  
Similarly, we can compute a nonnegative extended random variable $E^{(n)}= E^{(n)}(X^{(n)})$. The following definitions give some intuitive conditions on $P^{(n)}$ and $E^{(n)}$ for them to be considered as p-variables and e-variables in an asymptotic sense.

\begin{definition}[Asymptotic p-variables and e-variables]
\label{defi:asymp_evariables_pvariables}
A sequence $(E^{(n)})_{n\in\N}$ of $[0,\infty]$-valued random variables is called a sequence of:
\begin{itemize}
    \item  \emph{asymptotic} e-variables for $\mathcal P$,  if  for each $n$, $E^{(n)}$ is $(\varepsilon_n, \delta_n)$-approximate for $\mathcal P$ 
    for some $(\varepsilon_n, \delta_n)$, satisfying  
    \begin{equation}\label{eq:asymp-e} \mbox{for each $\p\in \mathcal P$,}\quad  \lim_{n\to \infty}  \varepsilon_n(\p) = \lim_{n\to \infty} \delta_n(\p) = 0;
    \end{equation}
    
    \item  \emph{strongly asymptotic} e-variables for $\mathcal P$, if condition \eqref{eq:asymp-e} is replaced by   
$$    \mbox{for each $\p\in \mathcal P$,} \quad \lim_{n\to \infty}   \varepsilon_n(\p) = 0 \quad \mbox{~and~} \quad  
 \delta_n(\p) = 0 \mbox{~for $n$ large enough;}$$ 
 
    \item \emph{uniformly asymptotic} e-variables for $\mathcal P$, if 
    condition \eqref{eq:asymp-e} is replaced by  $$ 
    \lim_{n\to \infty} \sup_{\p \in \cP} \varepsilon_n(\p) = \lim_{n\to \infty} \sup_{\p \in \cP} \delta_n(\p) = 0;$$     
    \item  \emph{uniformly strongly asymptotic} e-variables for $\mathcal P$, if condition \eqref{eq:asymp-e} is replaced by $$\lim_{n\to \infty} \sup_{\p \in \cP} \varepsilon_n(\p) = 0 \quad \mbox{~and~} \quad  \sup_{\p \in \cP} 
 \delta_n(\p) = 0 \mbox{~for $n$ large enough;}$$

\end{itemize}
 Sequences $(P^{(n)})_{n\in\N}$ of  (uniformly, strongly) asymptotic p-variables for $\mathcal P$ are defined analogously.

\end{definition}

We would also say ``asymptotic e-values'' (or p-values) in non-mathematical statements, and they can be interpreted as asymptotic e-variables (or p-variables) or their realizations. 
We sometimes   simply say that a sequence is asymptotic for $\cP$, and it should be clear from context whether we meant  a sequence of asymptotic e-variables ($[0,\infty]$-valued) or that of asymptotic p-variables  ($\R_+$-valued).\footnote{
\citet{liu1997notions} use the terminology ``limiting p-values'' for asymptotic p-values and ``limiting p-values in the strong sense'' for uniformly asymptotic p-values.
}
Clearly, if $\mathcal P$ is finite, then asymptotic e-variables or p-variables are also uniformly asymptotic.

We now study conditions and properties of asymptotic p-variables and e-variables. Our first result offers intuition regarding the formulation of strongly asymptotic e-variables.

\begin{proposition}
\label{prop:asym-e}
The sequence  $(E^{(n)})_{n\in\N}$ of   $[0,\infty]$-valued random variables is uniformly strongly asymptotic  for $\mathcal P$ if and only if $$\limsup_{n \to \infty} \sup_{\p\in \cP} \E^\p[E^{(n)}]\le 1 .$$ Similarly, it is   strongly asymptotic  for $\mathcal{P}$ if and only if $\limsup_{n \to \infty} \E^\p[E^{(n)}]\le 1$ for all $\p\in  \cP$.
\end{proposition}
The result below connects asymptotic e-variables and strongly asymptotic e-variables 
via uniform integrability. 

\begin{proposition}
\label{proposition:uniform_integrability_strong_e}
If a sequence of asymptotic e-variables is  uniformly integrable for each $\p\in \cP$, then it is  strongly asymptotic.
\end{proposition}

The next proposition gives a sufficient condition that best describes the intuition behind  asymptotic e-variables and p-variables. 
\begin{proposition}
\label{prop:converge-dist}
    If under each $\p\in\cP$, a nonnegative sequence  
converges to an e-variable (resp.~a p-variable) in distribution, then it
is a sequence of asymptotic e-variables  (resp.~p-variables) for~$\cP$.  
\end{proposition}

\begin{remark}
   \cite{gasparin2024combining} referred to asymptotic p-variables as those that converge to a p-variable in distribution. As Proposition~\ref{prop:converge-dist} illustrates, this is a sufficient condition for our definition of asymptotic p-values.
Our definition is more general. For instance, any sequence that is larger than a p-variable is a  sequence of asymptotic p-variables, and it does not need to converge. 
\end{remark}

The next result gives equivalent conditions for uniformly asymptotic p-variables. 
The proof is similar to that of Proposition~\ref{prop:converge-dist} and omitted here. 
\begin{proposition}
\label{prop:asym-p}
The sequence $(P^{(n)})_{n\in\N}$  of 
nonnegative random variables 
is uniformly asymptotic  for $\mathcal P$ if and only if  
$\limsup_{n\to\infty}\sup_{\p\in \cP} \p(P^{(n)}\le t)\le t $ for all $t\in (0,1)$.
Similarly, it
is asymptotic  for $\mathcal P$ if and only if  
$\limsup_{n\to\infty}\p(P^{(n)}\le t)\le t $ for all $\p\in \mathcal P$ and $t\in (0,1)$.
\end{proposition}

The most common way of constructing asymptotic p-variables and e-variables is via limit theorems such as the central limit theorem. The following example is a consequence of Propositions~\ref{proposition:uniform_integrability_strong_e} and~\ref{prop:converge-dist}.
\begin{example}
\label{example:weak_convergence_to_asymptotic}
Let $Z^{(n)}$ be a statistic computed based on $X^{(n)}$ for $n\in \N$. Suppose that
$$
Z^{(n)} \cd Z\, \text{ as }\, n \to \infty \,\text{ under every }\, \p \in \cP,
$$
where $Z $ is a random variable with a continuous cumulative distribution function $F$ and ``$\cd$'' denotes  convergence in distribution. Then the following hold:
\begin{enumerate}[label=(\roman*)]
\item Let $P^{(n)} = F(Z^{(n)})$ or $P^{(n)}=1-F(Z^{(n)})$ for $n\in \N$. Then $(P^{(n)})_{n\in\N}$ is a sequence of asymptotic p-variables for $\cP$.
\item Let $h: \mathbb R \to [0,\infty]$ be a function with $\E[h(Z)] \leq 1$ that is increasing, decreasing, or upper semi-continuous, and  $E^{(n)} = h(Z^{(n)})$ for $n\in \N$. Then $(E^{(n)})_{n\in\N}$ is a sequence of asymptotic e-variables for $\cP$.  
If $h$ is further bounded, then $(E^{(n)})_{n\in\N}$ is strongly asymptotic  for $\cP$.  
\end{enumerate}
\end{example}

The claim for asymptotic p-variables of course is standard and underlies most constructions of p-values used in practice. But the statement for e-variables is new, which is our main motivation for stating the above result.

The most common cases in practice include the following textbook examples~\citep{vandervaart1998asymptotic, lehmann2005testing} (also see Section~\ref{subsec:prototypical_single_asym} for an explicit example):
\begin{enumerate}
\item $Z \sim \mathrm{N}(0,1)$ as would follow from the central limit theorem.
\item $Z \sim |\mathrm{N}(0,1)|$, where we use the notation $ |\mathrm{N}(0,1)|$ to denote the half-normal distribution. This would follow also from the central limit theorem and the continuous mapping theorem.\footnote{An application of this case is as follows. In the setting of Example~\ref{example:weak_convergence_to_asymptotic}(ii), suppose $Z \sim \mathrm{N}(0,1)$ but we seek to use a function $h$ that is increasing in $|z|$ (rather than $z$) and that satisfies $\E[h(|Z|
)]\leq 1$. Then the sequence $(E^{(n)})_{n \in \mathbb N}$ with $E^{(n)}=h(|Z^{(n)}|)$ is also a sequence of asymptotic e-variables.
}
\item $Z \sim \chi^2_{\nu}$, where $\chi^2_{\nu}$ denotes the chi-square distribution with $\nu$ degrees of freedom. Such a limit occurs e.g., for likelihood ratio or chi-square tests.
\end{enumerate}

We end this subsection by noting that we can naturally calibrate asymptotic p-variables and e-variables by applying a (possibly different) calibrator to each element in the sequence, in the sense of Theorem~\ref{th:calibration}. 
\begin{proposition}
\label{prop:asymptotic_calibration}
    Calibrating a sequence of (uniformly, strongly) asymptotic p-variables  yields a sequence of (uniformly, strongly) asymptotic e-variables, and vice versa. 
\end{proposition} 
This result follows directly from Theorem~\ref{th:calibration}.

\subsection{Approximate compound p-values and e-values}
\label{subsec:approx_compound}

Just as in Definition~\ref{defi:approx} we introduced approximate p-variables and e-variables, below we define approximate compound p-variables and e-variables. This extends our earlier framework and builds upon the notion of approximate compound e-variables that was implicitly introduced in~\citet{ignatiadis2024evalues}. Throughout this section, we assume that each $\p$ is atomless; this means that we can simulate a uniformly distributed random variable under $\p$. This enables us to formulate the definition of approximate compound p-variables and e-variables following the equivalent formulation for approximate p-variables/e-variables in Proposition~\ref{prop:alter-p}.

\begin{definition}[Approximate compound p-variables and e-variables]
\label{defi:approx_compound}
Fix the null hypotheses $\cP_1,\dots,\cP_K$, a set of distributions $\cP$, and two functions $\varepsilon:\cP \to \R_+ $ and $\delta:\cP \to [0,1]$, and write $\varepsilon_\p=\varepsilon(\p)$
and $\delta_\p=\delta(\p)$ for $\p \in \cP$. 
\begin{enumerate}[label=(\roman*)]
    \item 
Nonnegative random variables $P_1,\dots,P_K$ are \emph{$(\varepsilon, \delta)$-approximate compound} p-variables for $(\cP_1,\dots,\cP_K)$ under $\cP$ if for every $\p \in \cP$,  there exists an event $A$ such that
$$\sum_{k: \p \in \cP_k}  \p(P_k\le t , A) \le  K(1+\varepsilon_\p )t ~\mbox{for all $t\in (0,1)$ ~~and~~ }\p(A) \geq 1-\delta_\p;$$
\item  $[0,\infty]$-valued random variables $E_1,\dots,E_K$ are \emph{$(\varepsilon, \delta)$-approximate compound} e-variables for $(\cP_1,\dots,\cP_K)$ under $\cP$ if  for every $\p \in \cP$, 
there exists an event $A$ such that
$$ \sum_{k: \p \in \cP_k}  \E^{\p}[E_k\id_{A}] \leq K(1+\varepsilon_\p) \mbox{  ~~and~~ }\p(A) \geq 1-\delta_\p.$$
\end{enumerate}
\end{definition}
We provide examples of approximate compound e-values further below. One prominent example in the recent literature is due to~\citet{lee2024boosting} who construct $(0,\delta)$-approximate compound e-variables (without calling them as such). In their application the $\delta$ inflation is due to Monte-Carlo error in approximating compound e-values and $\delta$ can be made arbitrarily close to $0$ by increased computational effort.

If we start with per-hypothesis approximate p-variables (or e-variables), then we recover approximate compound p-variables (or e-variables) for the multiple testing problem.
\begin{proposition}
\label{proposition:approximate_evalues_to_approximate_compound}
Suppose that for all $k \in \mathcal{K}$, $P_k$ (resp.~$E_k$) is an $(\varepsilon_k,\delta_k)$-approximate p-variable (resp.~e-variable) for $\cP_k$. Then, $P_1,\dots,P_K$ (resp.~$E_1,\dots,E_K$) are $(\varepsilon,\delta)$-approximate compound p-variables (resp.~e-variables) for the choice:
$$
\varepsilon(\p) = \frac{1}{K} \sum_{k:\p \in \cP_k} \varepsilon_k(\p),\qquad \delta(\p) = \sum_{k:\p \in \cP_k} \delta_k(\p).
$$  
\end{proposition}
We next provide an analogous statement to that of Proposition~\ref{prop:alter-p} for approximate compound e-variables.
A corresponding formulation for approximate compound p-variables is omitted.

\begin{proposition}
\label{prop:equiv_approx_compound_e}
Let $\varepsilon:\cP \to \R_+ $ and $\delta:\cP \to [0,1]$ and $E_1,\dots,E_K$ be $[0,\infty]$-valued random variables. If each $\p\in \cP$ is atomless, then Definition~\ref{defi:approx_compound}(ii) is equivalent to:
$$
\E^\p \left[\left(\sum_{k: \p \in \cP_k}E_k \right)\wedge t \right] \le  K(1+\varepsilon_\p)  +\delta_\p t  \qquad \mbox{for all $t\in \R_+$ and all $\p\in \mathcal P$}.
$$
\end{proposition}

We can calibrate approximate compound p-variables to approximate compound e-variables and vice versa as per usual if we use \emph{the same} calibrator for all hypotheses. For instance, given a p-to-e calibrator $h$, if $P_1,\dots,P_K$ are approximate compound p-variables for $\cP_1,\dots,\cP_K$ under $\cP$, then $h(P_1),\dots,h(P_k)$ are approximate compound e-variables for the same hypotheses.

\begin{proposition}
\label{th:calibration_approximate_compound}
 Let $\varepsilon:\cP \to \R_+ $ and $\delta:\cP \to [0,1]$. 
    Calibrating $(\varepsilon, \delta)$-approximate compound p-variables  yields  $(\varepsilon, \delta)$-approximate compound e-variables, and vice versa.
\end{proposition}

We end this subsection by introducing an alternative definition in lieu of Definition~\ref{defi:approx_compound} for approximate compound p-variables and e-variables. This definition has the advantage that it is easier to satisfy, but it provides weaker guarantees. We include it because the asymptotic notion it leads to for compound p-variables was central to the work of~\citet{ignatiadis2024empirical}. However, we emphasize that in the remainder of the main text, we focus on Definition~\ref{defi:approx_compound}, since, as we will see in Section~\ref{sec:multiple_testing}, approximate compound e-values as per Definition~\ref{defi:approx_compound} are naturally compatible with approximate false discovery rate guarantees for the e-BH procedure.

\begin{definition}[*Approximate\footnote{We use the asterisk (*approximate) to distinguish between two related but distinct definitions within this paper. We expect other authors may reasonably use ``$(\varepsilon,\delta)$-approximate compound p-variables'' or ``$(\varepsilon,\delta)$-approximate compound e-variables'' for either notion, provided the definition is clearly stated.} compound p-variables and e-variables]
\label{defi:additive_approx_compound}
Fix the null hypotheses $\cP_1,\dots,\cP_K$, a set of distributions $\cP$, and two functions $\varepsilon:\cP \to \R_+ $ and $\delta:\cP \to [0,1]$, and write $\varepsilon_\p=\varepsilon(\p)$
and $\delta_\p=\delta(\p)$ for $\p \in \cP$. 
\begin{enumerate}[label=(\roman*)]
    \item 
Nonnegative random variables $P_1,\dots,P_K$ are \emph{$(\varepsilon, \delta)^*$-approximate compound} p-variables for $(\cP_1,\dots,\cP_K)$ under $\cP$ if 
$$\sum_{k: \p \in \cP_k} \p(P_k\le t) \le  K ( (1+\varepsilon_\p )t +\delta_\p ) \qquad \mbox{for all $t\in (0,1)$ and all $\p\in \mathcal P$.} $$
\item $[0,\infty]$-valued random variables $E_1,\dots,E_K$ are \emph{$(\varepsilon, \delta)^*$-approximate compound} e-variables for $(\cP_1,\dots,\cP_K)$ under $\cP$ if  
 $$\sum_{k: \p \in \cP_k} 
 \E^\p [E_k\wedge t] \le  K(1+\varepsilon_\p  +\delta_\p t)  \qquad \mbox{for all $t\in \R_+$ and all $\p\in \mathcal P$.}$$ 
\end{enumerate}
\end{definition}

When all $\p \in \cP$ are atomless and $K=1$, then Definitions~\ref{defi:additive_approx_compound} and~\ref{defi:approx_compound} are equivalent to each other and equivalent to Definition~\ref{defi:approx} on approximate p-variables and e-variables. The definitions also agree when $\delta=0$, for instance, $(\varepsilon,0)^*$-approximate compound e-variables are identical to $(\varepsilon,0)$-approximate compound e-variables. Definition~\ref{defi:additive_approx_compound} is implied by Definition~\ref{defi:approx_compound}, as can be seen most clearly for approximate compound e-variables by comparing to the result of Proposition~\ref{prop:equiv_approx_compound_e}.
In Appendix~\ref{sec:alt_defi} we provide further results regarding Definition~\ref{defi:additive_approx_compound}.

\subsection{Asymptotic compound p-values and -values}
\label{subsec:asymp_compound_evalues}

In this subsection we define asymptotic compound p-variables and e-variables. The reader may directly anticipate the definitions that follow. However, what is most important, is to clarify the asymptotic setup.
As before, we seek to test hypotheses $\mathcal{P}_1,\dots,\mathcal{P}_K$ based on data $X^{(n)}$. 
We are interested in asymptotics in which either $n \to \infty$ (with $K$ fixed) or $K \to \infty$ (with $n$ fixed) or both $K,n \to \infty$. 

Formally, we consider triangular-array type of asymptotics indexed by a single parameter $m \in \mathbb N$ that determines the dataset size, $n=n(m)$, and the number of hypotheses, $K=K(m)$. We require that $m \to \infty$ and that $\max\{K, n\} \to \infty$ as $m \to \infty$.

All three asymptotic regimes are frequently encountered in the (multiple) testing literature. The asymptotic regime with $n \to \infty$ and $K$ fixed is the setting we considered in Section~\ref{sec:asymp-e} (with $K=1$), while~\citet{dudoit2008multiple} provide a textbook treatment of multiple testing in the regime with $n \to \infty$ and $K>1$ fixed. When using empirical Bayes arguments to share strength and learn nuisance parameters across hypotheses, one is relying on asymptotics as $K$ grows to infinity (as in e.g., the seminal FDR asymptotics in~\citet{storey2004strong}) and often with $n$ fixed (see~\citet{ignatiadis2024empirical} for a concrete example). Finally, other authors, e.g.,~\citet{cao2011simultaneous, liu2014phase}, have studied multiple testing procedures in the regime where both $n, K \to \infty$.

\begin{definition}[Asymptotic compound p-variables and e-variables]
\label{defi:asymp_compound_evariables_pvariables}
A sequence 
$$\left( E^{(n(m))}_1,\dots, E^{(n(m))}_{K(m)}\right)_{m\in\N}$$ 
of tuples of $[0,\infty]$-valued random variables is called a sequence of:
\begin{itemize}
    \item  \emph{asymptotic compound}  e-variables for $\cP_1,\dots,\cP_{K(m)}$ under $\cP$,  if  for each $m$, $E^{(n(m))}_1,\dots, E^{(n(m))}_{K(m)}$ are $(\varepsilon_m, \delta_m)$-approximate compound e-variables for $\cP_1,\dots,\cP_{K(m)}$ under $\cP$
    for some $(\varepsilon_m, \delta_m)$, satisfying  
    \begin{equation}\label{eq:asymp-compound-e} \mbox{for each $\p\in \mathcal P$,}\quad  \lim_{m\to \infty}  \varepsilon_m(\p) = \lim_{m\to \infty} \delta_m(\p) = 0;
    \end{equation}
    
    \item  \emph{strongly asymptotic compound} e-variables for $\cP_1,\dots,\cP_{K(m)}$ under $\cP$, if condition \eqref{eq:asymp-compound-e}  is replaced by   
$$    \mbox{for each $\p\in \mathcal P$,} \quad \lim_{m\to \infty}   \varepsilon_m(\p) = 0 \quad \mbox{~and~} \quad  
 \delta_m(\p) = 0 \mbox{~for $m$ large enough.}$$ 
\end{itemize}
Uniformly (strongly) asymptotic compound e-variables are defined analogously. Moreover, sequences 
$$\left( P^{(n(m))}_1,\dots, P^{(n(m))}_{K(m)}\right)_{m\in\N}$$ 
of (uniformly, strongly) asymptotic compound p-variables are defined analogously. 
\end{definition}

We provide explicit constructions of asymptotic compound e-values in Sections~\ref{subsec:odp_ebayes} and~\ref{subsec:sum_of_squares_compound}.

\section{Compound e-values are central to multiple testing}
\label{sec:multiple_testing}
Our goal in this section is to demonstrate that compound e-values are fundamental to multiple testing and provide a useful conceptual framework for developing new multiple testing procedures that control the false discovery rate (FDR). To motivate what follows, we first record an observation about single hypothesis testing in which we seek to construct a level-$\alpha$ non-randomized test for a hypothesis $\cP$, that is, a $\{0,1\}$-valued function $\phi(\cdot)$ such that $\mathbb E^{\p}[\phi(X)] \leq \alpha$ for all $\p \in \mathcal{P}$. First, if $E=E(X)$ is an e-value for $\cP$, then $\phi(X) = \id_{\{E(X) \geq 1/\alpha\}}$ is a level-$\alpha$ test, since by Markov's inequality $\p(E(X) \geq 1/\alpha) \leq  \alpha \E^{\p}[E(X)] \leq \alpha$ for all $\p \in \mathcal{P}$. Moreover, given any level-$\alpha$ test $\phi(X)$, we can construct an e-value via $E=\phi(X)/\alpha$ in which case $E \geq 1/\alpha$ if and only if $\phi(X)=1$. Thus, every level-$\alpha$ test (including those based on thresholding p-values) implicitly uses an e-value and is recoverable by thresholding that e-value at $1/\alpha$~\citep[Chapter 2]{ramdas2024hypothesis}.

The remainder of this section develops a generalization of the above relation in the context of multiple testing. Our primary result is that every FDR procedure can be recovered by applying the e-BH procedure of~\citet{wang2022false} to compound e-values.

We first introduce some further notation.
We denote a multiple testing procedure by $\mathcal{D}$, that is, a Borel function of $X$  that produces a subset of $\mathcal K$
representing the indices of rejected hypotheses. The rejected hypotheses by  $\mathcal{D}$ are called discoveries. We write $V_k = \id_{\{k \in \mathcal{D}\}}$ as the indicator of whether $\cP_k$ is rejected by $\mathcal{D}$,
 $F_{\mathcal{D}} = F_{\mathcal{D}}(\p)$ as the number of true  null hypotheses that are rejected (i.e., false discoveries) when the true distribution is $\p$: 
\[
F_{\cD} = \sum_{k: \p \in \mathcal{P}_k} V_k,
\]
and $R_{\mathcal{D}}$ as the total number of discoveries:
\[
R_{\cD} = \sum_{k \in \cK} V_k.
\]
\citet{benjamini1995controlling}
proposed to control the FDR, which is the expected value of the false discovery proportion,
that is, $ 
\mathrm{FDR}_{\cD}^{\p}:=\E^{\p}[{F_{\cD}}/{R_{\cD} }]$ with the convention $0/0=0$. 

Throughout this section (just as in Section~\ref{subsec:approx_compound}), we assume that each $\p$ is atomless; this means that we can simulate a uniformly distributed random variable under $\p$. This allows us to use the equivalent definition in Proposition~\ref{prop:alter-p}  of approximate e-variables through an event $A$ with a constraint on its probability.
  
\begin{definition}[FDR control]
Fix the null hypotheses $(\cP_1,\dots,\cP_K)$ and a set $\cP $ of distributions.
Let $\mathcal{D}$ be a multiple testing procedure and let $\varepsilon:\cP \to \R_+ $ and $\delta:\cP \to [0,1]$. The procedure $\mathcal{D}$ is said to have
\begin{itemize}
\item FDR \emph{control} at level $\alpha$ for $(\mathcal{P}_1,\dotsc,\mathcal{P}_K)$ under $\cP$ if for every $\p \in \cP$
$$ \mathrm{FDR}_{\cD}^{\p} = \E^{\p}\left[\frac{F_{\cD}}{R_{\cD}}\right] \leq \alpha;$$
\item \emph{$(\varepsilon,\delta)$-approximate} FDR \emph{control} at level $\alpha$ for $(\mathcal{P}_1,\dotsc,\mathcal{P}_K)$ under $\cP$, if for every $\p \in \cP$, 
there exists an event $A$ such that
$$\E^{\p}\left[\frac{F_{\cD}}{R_{\cD}} \id_{A}\right] \le  \alpha(1+\varepsilon_\p ) ~\mbox{~~and~~ }\p(A) \geq 1-\delta_\p;$$
\item \emph{asymptotic} FDR \emph{control} at level $\alpha$
for $(\cP_1,\dots,\cP_{K(m)})$ under $\cP$ if for each $m \in \mathbb N$ (with $m$ indexing a sequence of problems as in Section~\ref{subsec:asymp_compound_evalues}) it has $(\varepsilon_m, \delta_m)$-approximate FDR control at level $\alpha$
for some $(\varepsilon_m, \delta_m)$, satisfying  
$$
\mbox{for each $\p\in \mathcal P$,}\quad  \lim_{m\to \infty}  \varepsilon_m(\p) = \lim_{m\to \infty} \delta_m(\p) = 0.
$$
\end{itemize}
\end{definition}
We make a few remarks. A notion related to approximate FDR control was used as proof technique (without being given that name) by~\citet{blanchard2008two}.

A procedure with $(\varepsilon, \delta)$-approximate FDR control also satisfies $$\mathrm{FDR}_{\cD}^{\p} \leq \alpha(1+\varepsilon_{\p}) + \delta_{\p} \text{ for any } \p \in \cP.$$
Conversely, any procedure satisfying the above inequality 
has $(\epsilon,\delta/\alpha)$-approximate FDR control, by invoking the equivalence in Proposition \ref{prop:alter-p}; more precisely,
the above condition implies, for all $t\ge 1$,
$$\E^{\p}\left[\frac{F_{\cD}}{\alpha {R_{\cD}}} \wedge t\right] \le\E^{\p}\left[\frac{F_{\cD}}{\alpha {R_{\cD}}} \right] 
\le (1+\epsilon_\p) + \frac{\delta_\p}{\alpha} \le(1+\epsilon_\p) + \frac{\delta_\p}{\alpha} t . $$

Let $E_1,\dots,E_K$ be compound e-values. Then we can control the FDR through the e-BH procedure, which is defined as follows.
\begin{definition}[e-BH procedure, \citealp{wang2022false}]
\label{defi:eBH}
Let  $E_1,\dots, E_K$ be $[0,\infty]$-valued random variables.
For $k\in \mathcal K$, let $E_{[k]}$ be the $k$-th order statistic of $E_1,\ldots,E_K$, from the largest to the smallest.  The e-BH procedure rejects all hypotheses with the largest $k_e^*$ e-values, where 
\begin{equation*} 
\label{eq:e-k-intro} 
k_e^*:=\max\left\{k\in \mathcal K: \frac{k E_{[k]}}{K} \ge \frac{1}{\alpha}\right\},
\end{equation*}   
with the convention $\max(\varnothing) = 0$.
\end{definition} 
The following theorem records the fact that e-BH controls the FDR.
\begin{theorem}
\label{theo:eBH_controls_the_FDR}
Suppose that we apply the e-BH procedure at level $\alpha$ to $E_1,\dots,E_K$. Then:
\begin{itemize}
\item If $E_1,\dots,E_K$ are compound e-values, then e-BH controls the FDR at level $\alpha$.
\item If $E_1,\dots,E_K$ are $(\varepsilon, \delta)$-approximate compound e-values, then e-BH satisfies $(\varepsilon, \delta)$-approximate FDR control, and so it controls
the FDR at level $\alpha(1+\varepsilon) + \delta$. 
\item If $E_1,\dots,E_K$ are asymptotic compound e-values, then e-BH has asymptotic FDR control at level $\alpha $.
\end{itemize}
\end{theorem}
Notably, the above result requires no assumption whatsoever on the dependence across the $E_k$. 
The results in~\citet{wang2022false} handled $(\varepsilon, 0)$-approximate compound e-values, so  Theorem~\ref{theo:eBH_controls_the_FDR} offers a slight generalization in that respect. 

Several authors including \citet{ren2024derandomised, li2025note} have observed that specific multiple testing procedures may be recast as e-BH with a specific choice of compound e-values. In fact one can show, as we do next, that this holds for any multiple testing procedure that controls the FDR at a known level. This implies that if one seeks to control the FDR, it is without loss of generality to search among all procedures that apply e-BH to compound e-values.  

In what follows, an FDR procedure $\mathcal D$ at level $\alpha $ is admissible if 
for any FDR procedure $\mathcal D'$ at level $\alpha $ such that $\mathcal D \subseteq \mathcal D'$, we have $
\mathcal D=\mathcal D'$.

\begin{theorem}[Universality of e-BH with compound e-values]
\label{theo:universality_eBH}
Let $\cD$ be any procedure that controls the FDR at a known level $\alpha$ for the hypotheses $\cP_1,\ldots,\cP_K$ under $\cP$. Then, there exists a choice of compound e-values such that the e-BH procedure yields identical discoveries as $\cD$. 
Further, if the original FDR procedure is admissible, then these compound e-values can be chosen to be tight. 
\end{theorem}

\begin{proof}
We first construct compound e-values from $\cD$ following~\citet{banerjee2023harnessing}.  Let
\begin{align}\label{eq:universal-e} E_k = \frac{K}{\alpha}\frac{V_k}{R_{\cD} \lor 1}.
\end{align}
Fix any $\p \in \cP$. Then,
$$
\sum_{k: \p \in \mathcal{P}_k} \mathbb E^{\p}[ E_k] = \sum_{k: \p \in \mathcal{P}_k}\mathbb E^{\p}\left[ \frac{K}{\alpha}\frac{V_k}{R_{\cD} \lor 1}\right] = \frac{K}{\alpha}\mathrm{FDR}_{\cD}^{\p} \leq \frac{K}{\alpha}\alpha = K,$$
where the inequality follows from the definition of FDR guarantee. Thus, $E_1,\dots, E_K$ are indeed compound e-variables for $\mathcal{P}_1,\dotsc,\mathcal{P}_K$ under $\cP$, 

By Theorem~\ref{theo:eBH_controls_the_FDR}, e-BH applied to $E_1,\ldots,E_K$ controls the FDR. Moreover, $k_e^*$ in Definition~\ref{defi:eBH} is seen to be equal to $R_{\cD}$. Thus, $\cD$ and e-BH on $E_1,\ldots,E_K$ make the same discoveries.
We postpone the proof of tightness to Appendix~\ref{sec:ebh_admissible_proof}.
\end{proof}

Such a result is not known to hold for the BH procedure, or for any other procedure. The e-BH procedure, coupled with the notion of compound e-values, truly stands out in that respect.

Theorem~\ref{theo:universality_eBH} captures the fundamental role of compound e-values for multiple testing with FDR  control. There are several further practical applications in which compound e-values play a fundamental role in multiple testing. We review four such applications in Appendix~\ref{sec:compound_mtp_addendum}:

\begin{itemize}[leftmargin=*, wide, labelwidth=!, labelindent=0pt]
    \item Appendix~\ref{subsec:derandomization} explains how compound e-values enable combination and derandomization of multiple testing procedures, including a new application to derandomization of the randomized e-BH procedures of~\citet{xu2024more}.
    \item Appendix~\ref{subsec:compound_evalues_as_weights} demonstrates that compound e-values provide a natural notion of weights for the ep-BH procedure of~\citet{ignatiadis2024evalues}.
    \item Appendix~\ref{sec:compound_evalues_to_merge} shows that compound e-values can be used as weights for merging p-values. These results are new. 
    \item Appendix~\ref{subsec:compound_pvalues_multiple_testing} briefly describes the role of compound p-values in multiple testing. Compound p-values can be used for multiple testing via calibration to compound e-values and applying the e-BH procedure; this argument generalizes a result of~\citet{armstrong2022false} and provides an alternative proof that the Benjamini-Yekutieli procedure applied to compound p-values controls the FDR. 
\end{itemize}

\section{Sequence models and compound decisions}
\label{sec:sequence_model}  

In this section we focus on sequence models, wherein the data $X$ may be written as $X=(X_k : k \in \mathcal{K})$ and in principle we may test each hypothesis $\cP_k$ using only $X_k$. 
We record the following definition, which we will call upon throughout this section.
\begin{definition}[Simple separable e-values]
\label{defi:simple_separable}
$E_1,\dots,E_K$ are called separable if for all $k$, $E_k$ is $X_k$-measurable, that is $E_k = E_k(X_k)$. They are called simple separable if $E_k = f(X_k)$ for some function $f$ (which is the same for all $k$).
\end{definition}
In what follows, we provide a brief summary of compound decision theory which motivates the nomenclature ``compound e-values'' (Section~\ref{subsec:compound}) and then we provide two constructions of (approximate) compound e-values motivated by compound decision theory (Sections~\ref{subsec:best_simple_separable}  and~\ref{subsec:compound_ui}).

\subsection{Compound decision theory}
\label{subsec:compound}
Why do we use the term compound e-values? Our motivation is to pay tribute and connect the definition to Robbins' compound decision theory~\citep{robbins1951asymptotically} in which multiple statistical problems are connected through a (compound) loss function that averages over individual losses. We refer the reader to e.g.~\citet{copas1969compound, zhang2003compound, jiang2009general} for more comprehensive accounts and present  a brief overview here. As a very concrete example (which was studied in detail already in~\citet{robbins1951asymptotically}), consider the Gaussian sequence model:
\begin{equation}
\label{eq:gaussian_sequence}
\boldmu = (\mu_1,\dotsc,\mu_K) \text{ fixed},\quad X_k \simindep \mathrm{N}(\mu_k, 1) \text{ for } k \in \mathcal{K}.
\end{equation}
Robbins was interested in constructing estimators $\hat{\mu}_k$ of $\mu_k$ such that the following expected compound loss would be small:
\begin{equation}
\label{eq:compound_risk}
\frac{1}{K}\sum_{k \in \mathcal{K}} \E[(\hat{\mu}_k - \mu_k)^2].
\end{equation}
If $\mu_k \simiid M$ for some prior distribution $M$, then the best estimator is the Bayes estimator $\hat{\mu}_k^{B} = \E^M{[\mu_k \mid X]} = \E^M{[\mu_k \mid X_k]}$, which is simple separable (in the sense of Definition~\ref{defi:simple_separable}).\footnote{The conditional expectations in the above displays refer to the joint distribution of $(\mu_k, X_k),\, k \in \mathcal{K}$.} Robbins 
asked: Suppose $\boldmu$ is deterministic as in~\eqref{eq:gaussian_sequence}. Given knowledge of $\boldmu$,
which function $s_{\boldmu}: \R \to \R$ leads to a simple separable estimator $\hat{\mu}_k^{s_{\boldmu}} = s_{\boldmu}(X_k)$ that minimizes the risk in~\eqref{eq:compound_risk}? The answer lies in the fundamental theorem of compound decisions, which formally connects~\eqref{eq:gaussian_sequence} to the univariate Bayesian problem with prior $M(\boldmu) = \sum_{k \in \mathcal{K}} \delta_{\mu_k}/K$ equal to the empirical distribution of $\boldmu$ (with $\delta_{\mu_k}$ denoting the Dirac measure at $\mu_k$):
\begin{equation}
\label{eq:bayes_compound}
\mu' \sim M(\boldmu) ,\;\; X' \mid \mu' \sim \mathrm{N}(\mu', 1). 
\end{equation}
The fundamental theorem of compound decision theory states the following. Given any fixed  $s: \R \to \R$, we have that
\begin{equation}
\label{eq:fundamental_compound_decisions}
\frac{1}{K}\sum_{k \in \mathcal{K}} \E^{\boldmu}[(s(X_k) - \mu_k)^2]  = \sum_{k \in \mathcal{K}}\frac{1}{K}\E^{\mu_k}[(s(X_k) - \mu_k)^2] = \E^{M(\boldmu)}[(s(X') - \mu')^2].
\end{equation}
In the left hand side display, $\boldmu \in \R^{K}$ is treated as fixed (as in~\eqref{eq:gaussian_sequence}), while in the right hand side display, we only have a single random $\mu' \in \R$ that is randomly drawn from $M(\boldmu)$ (as in~\eqref{eq:bayes_compound}). From~\eqref{eq:fundamental_compound_decisions} it immediately follows that the optimal simple separable estimator is the Bayes estimator under $M(\boldmu)$, that is $s(x) = \E^{M(\boldmu)}[\mu' \mid X'=x]$.
This construction motivated Robbins to construct feasible non-separable estimators $\hat{\mu}_k = \hat{\mu}_k(X)$ that have risk close to the optimal simple separable estimator.\footnote{Thus the optimal simple separable estimator defines an oracle benchmark. Just as in classical decision theory, one often restricts attention to subclasses of estimators, e.g., equivariant or unbiased, Robbins set a benchmark defined by simple separable estimators. What is slightly unconventional here is a fundamental asymmetry: the oracle estimator must be simple separable but has access to the true $\boldmu$, while the feasible estimator is non-separable but must work without knowledge of $\boldmu$.}
Empirical Bayes~\citep{robbins1956empirical} ideas work well for this task, for example,~\citet{jiang2009general} provide sharp guarantees for the performance of a nonparametric empirical Bayes method in mimicking the best simple separable estimator in the Gaussian sequence model in~\eqref{eq:gaussian_sequence} as $K \to \infty$.

\subsection{The optimal discovery compound e-values}
\label{subsec:best_simple_separable}
The connection of compound e-values to the fundamental theorem of compound decisions is born out from averaging a suitable statistical criterion, e.g., replacing the criterion $\E^{\p}[E_k] \leq 1$ for all $\p \in \mathcal{P}_k$ by a corresponding requirement summed over all null $k$. 
To further clarify the connection, we provide a construction of optimal simple separable compound e-values in sequence models with dominated null marginals via the fundamental theorem of compound decisions.

Suppose that $X=(X_k : k \in \mathcal{K})$, where $X_k$ takes values in a space $\mathcal{X}$. For any $\p \in \mathcal{P}$, let $\p_k$ be the $k$-th marginal of $\p$, that is, the distribution of $X_k$ when $(X_k : k \in \mathcal{K}) \sim \p$.
Suppose that for all $k \in \mathcal{K}$, the $k$-th hypothesis
$\cP_k$ is specified as a simple point null hypothesis on the $k$-th marginal, i.e., $\mathcal{P}_k = \{\p\,:\, \p_k = \p_k^{\circ}\}$ for some prespecified distribution $\p_k^{\circ}$ with $\dd\nu$-density equal to $p_k^{\circ}$, where $\nu$ is a common dominating measure. We assume that there exists $\p \in \cP$ such that $\p_k = \p_k^{\circ}$ for all $k \in \mathcal{K}$ (as would be the case, e.g., if the $X_k$ are independent and we take the product measure $\p = \bigotimes_{k \in \mathcal{K}}  \p_k^{\circ}$).

For $k \in \mathcal{K}$, we let $\mathbb Q_k^{\circ}$ be distributions on $\mathcal{X}$ with $\dd\nu$-densities $q_k^{\circ}$. We seek to solve the following optimization problem:    
\begin{equation}
\label{eq:optim_compound}
    \begin{aligned}
    & \underset{s(\cdot)}{\text{maximize}} 
    & & \frac{1}{K}\sum_{k \in \mathcal{K}} \E^{\mathbb Q_k^{\circ}}[\log(s(X_k))] \\
    & \text{subject to}
    & & s: \mathcal{X} \to [0,\infty] \\
    & & & s(X_1),\dots,s(X_K) \text{ are } \text{compound e-values}.
    \end{aligned}
\end{equation}

\begin{theorem}[Optimal simple separable compound e-values]
\label{th:compound-optimal}
The optimal solution to optimization problem~\eqref{eq:optim_compound} is given by the likelihood ratio of mixtures over the alternative and the null:
\begin{equation}
\label{eq:optimal_evalue}
s(x) = \frac{\sum_{j \in \mathcal{K}} q_j^{\circ}(x)}{\sum_{ j \in \mathcal{K}} p_j^{\circ}(x)}.
\end{equation}
\end{theorem}

\begin{proof}
Let $\p^{\circ}$ be a probability measure with $k$-th marginal given by $\p_k^{\circ}$ for all $k \in \mathcal{K}$. We will first solve the optimization problem subject to the weaker constraint that $s(X_1),\dotsc,s(X_K)$ satisfy $\sum_{k \in \mathcal{K}} \E^{\p^{\circ}}[s(X_k)] \leq K$.

We will next provide an argument inspired by Robbins' fundamental theorem of compound decisions. Define the mixture distributions  $\overline{\p}^{\circ} = \sum_{k \in \mathcal{K}} \p_k^{\circ}/K$ and $\overline{\mathbb Q}^{\circ} = \sum_{k \in \mathcal{K}} \mathbb Q_k^{\circ}/K$. Then, for any fixed $s : \mathcal{X} \to [0,\infty]$, we have the following formal equalities:
$$
\frac{1}{K}\sum_{k \in \mathcal{K}} \E^{\p_k^{\circ}}[\log(s(X_k))] = \E^{X' \sim \overline{\p}^{\circ}}[\log(s(X'))],\;\;\; \frac{1}{K}\sum_{k \in \mathcal{K}} \E^{\mathbb Q_k^{\circ}}[s(X_k)] = \E^{X' \sim \overline{\mathbb Q}^{\circ}}[ s(X')].
$$
Thus, it suffices to solve the following optimization problem:
\begin{equation*}
    \begin{aligned}
    & \underset{s(\cdot)}{\text{maximize}} 
    & & \E^{X' \sim \overline{\mathbb Q}^{\circ}}[\log(s(X'))] \\
    & \text{subject to}
    & & s : \mathcal{X} \to [0,\infty] \\
    & & & \E^{X' \sim \overline{\p}^{\circ}}[s(X')] \leq 1.
\end{aligned}
\end{equation*}
This is a well-known optimization problem~\citep{kelly1956new} 
and has solution given by~\eqref{eq:optimal_evalue}. Finally we may verify that $s(\cdot)$ in~\eqref{eq:optimal_evalue} indeed yields compound e-values. To this end, write $\overline{q}^{\circ}$ for the $\dd \nu$ density of $\overline{\mathbb Q}^{\circ}$. Then, we have the following for any $\p$:
$$\sum_{k: \p_k = \p_k^{\circ} } \mathbb E^{\p}[s(X_k)] = \sum_{k: \p_k = \p_k^{\circ} }\int \frac{K \overline{q}^{\circ}(x)}{ \sum_{ j \in \mathcal{K}} p_j^{\circ}(x)}p_k^{\circ}(x)\nu(\dd x)\leq \int K \overline{q}^{\circ}(x) \nu(\dd x) = K,$$
and this completes the proof.
\end{proof}
The best simple separable compound e-values in~\eqref{eq:optimal_evalue} take precisely the form of the statistics of the optimal discovery procedure (ODP) of~\citet{storey2007optimal} and~\citet*{ storey2007optimala}. We briefly recall the  statistical guarantee of the ODP. We call a multiple testing procedure $\mathcal{D}$ simple and separable (cf. Definition~\ref{defi:simple_separable}) if there exists a function $\phi: \mathcal{X} \to \{0,1\}$ such that $\id_{\{k \in \mathcal{D}\}} = \phi(X_k)$ for all $k \in \mathcal{K}$.~\citet{storey2007optimal} proves the following.
\begin{proposition}[Neyman-Pearson optimality of ODP]
Let $s(\cdot)$ be defined as in~\eqref{eq:optimal_evalue}. Then, for any function $\phi:\mathcal{X} \to \{0,1\}$ and any $c>0$, the following implication holds:
$$\sum_{k \in \mathcal{K}} \mathbb E^{\p_k}[\phi(X_k)] \leq \sum_{k \in \mathcal{K}} \mathbb E^{\p_k}[\id_{\{s(X_k) \geq c\}}] \,\;\Longrightarrow \;\, 
\sum_{k \in \mathcal{K}} \mathbb E^{\mathbb Q_k}[\phi(X_k)] \leq \sum_{k \in \mathcal{K}} \mathbb E^{\mathbb Q_k}[\id_{\{s(X_k) \geq c\}}]. 
$$
\end{proposition}
The interpretation of the result is as follows~\citep[Section 5.2]{storey2007optimal}. Suppose for each $k \in \mathcal{K}$ we draw $X_k$ with probability $\pi_0 \in (0,1)$ from $\p_k$ and probability $1-\pi_0$ from $\mathbb Q_k$. Then the procedure that thresholds the optimal discovery statistics $s(X_k)$ maximizes the expected number of true discoveries among all simple separable multiple testing procedures with the same expected number of false discoveries. Together with Theorem~\ref{th:compound-optimal}, this motivates our terminology: we call $s(X_1),\ldots,s(X_k)$ the optimal discovery compound e-values.

Some remarks are in order:
\begin{itemize}[leftmargin=*, wide, labelwidth=!, labelindent=0pt]
\item The objective $\E[\log(s )]$ that we use in \eqref{eq:optim_compound} is a long-standing criterion of optimality; see the classic works of \cite{kelly1956new,breiman1961optimal,bell1980competitive},
and the more recent works of \cite{shafer2021testing,grunwald2024safe,vovk2024nonparametric,waudby2024estimating,larsson2025numeraire}.
This quantity represents the rate of growth of e-processes in a single-hypothesis setting. 
Other utility functions beyond logarithmic utility can be used in \eqref{eq:optim_compound} and lead to different forms of optimal simple separable compound e-values. Our framework provides a general method for translating single e-value results to the compound setting. In Appendix~\ref{sec:optimal_compound_general_utilities}, we illustrate this approach using results of~\citet{koning2025continuous} on optimal e-values with more general utility functions.
\item The objective value of the optimization problem is given by \smash{$\DKL{\overline{\mathbb Q}^{\circ}}{\overline{\p}^{\circ}}$}, where $\DKL{\cdot}{\cdot}$ is the Kullback-Leibler divergence. By convexity, \smash{$\DKL{\overline{\mathbb Q}^{\circ}}{\overline{\p}^{\circ}} \leq \sum_{k \in \mathcal{K}} \DKL{\mathbb Q_k^{\circ}}{\p_k^{\circ}}/K$}, which means that the optimal discovery compound e-values have a worse objective than using the optimal separable e-value for each individual testing problem of $\p_k^{\circ}$ vs $\mathbb Q_k^{\circ}$ given by $E_k^* = q_k^{\circ}(X_k)/p_k^{\circ}(X_k)$.
In view of the above, why are the optimal discovery compound e-values of interest? The reason is that if $K$ is large (but the individual statistical problems have limited sample sizes), then it becomes possible to mimic the optimal discovery compound e-values without exact knowledge of the null $\p_k^{\circ}$, e.g., via empirical Bayes methods. Such a strategy is fruitful when the hypotheses $\cP_k$ are in fact composite. We demonstrate such a strategy explicitly in Section~\ref{sec:ttest} on simultaneous t-tests where the nuisance parameters are the variances.
\end{itemize}

We finally note that any tight simple separable compound e-value for testing point nulls $\mathcal{P}_k = \{\p\,:\, \p_k = \p_k^\circ\}$ must be of the form~\eqref{eq:optimal_evalue} for some densities $q_1^\circ,\dots,q_K^\circ$.

\begin{proposition}
\label{prop:compound_separable}
Suppose that $E_1,\dots,E_K$ are tight simple separable compound e-values (with the nulls defined as in the previous parts of this section).  
Then there exist $\dd \nu$-densities $q_1^\circ,\dots,q_K^\circ$ such that $E_k$ may be represented as in~\eqref{eq:optimal_evalue} on the support of $\sum_{k \in \mathcal{K}} \p_k^\circ/K$ and it is without loss of generality to take $q_1^\circ=\dots=q_K^\circ$. 
\end{proposition}

\subsection{Localized and compound universal inference}
\label{subsec:compound_ui}

We now slightly modify the setting of Section~\ref{subsec:best_simple_separable} to accommodate composite null hypotheses with nuisance parameters. We continue to suppose that $X=(X_k : k \in \mathcal{K})$, where $X_k$ takes values in a space $\mathcal{X}$ and that the $k$-th hypothesis $\cP_k$ is specified solely in terms of the marginal distribution of $X_k$. However, we now suppose that $\cP_k$ is composite and that $k$-th null marginal may be parameterized by a nuisance parameter $\psi_k$ that lies in a parameter space $\Psi$. Concretely, $\mathcal{P}_k = \{\p: \dd \p_k/ \dd \nu = p^\circ_{\psi_k} \text{ for some } \psi_k \in \Psi\}$, where $\{p^\circ_{\psi} : \psi \in \Psi\}$ is a family of $\dd \nu$-densities on $\mathcal{X}$. We seek to test each composite null $\cP_k$ against the alternative $\{\mathbb Q: \dd \mathbb Q_k/ \dd \nu = q^\circ\}$, where $q^\circ$ is a further $\dd \nu$-density on $\mathcal{X}$ (which we choose to be the same for all $k \in \mathcal{K}$).

In the above setting, universal inference (UI) of~\citet{wasserman2020universal} provides a flexible and practical way of constructing separable e-values:\footnote{Throughout this section we assume for simplicity that infima are attained.}
\begin{equation}
E_k^{\mathrm{UI}} = \frac{ q^\circ(X_k)}{ p^\circ_{\widehat{\psi}_k^{\mathrm{UI}}}(X_k)},\,\,\,\;\widehat{\psi}_k^{\mathrm{UI}} \in \argmin \left\{ \frac{q^\circ(X_k)}{p^\circ_{\psi}(X_k)}\,:\, \psi \in \Psi \right\}. 
\label{eq:UI}
\end{equation}
It is easy to see that $E_k^{\mathrm{UI}}$ is a valid e-value. 

We now propose an alternative to the universal inference e-value above that we call localized universal inference (LUI) e-value. Let \smash{$\widehat{\Psi}_k(\delta) \subseteq\Psi$} be a $(1-\delta)$-confidence set for the nuisance parameter $\psi_k$. Then define:
\begin{equation}
E_k^{\mathrm{LUI}} = \frac{ q^\circ(X_k)}{ p^\circ_{\widehat{\psi}_k^{\mathrm{LUI}}}(X_k)},\,\,\, \; \widehat{\psi}_k^{\mathrm{LUI}} \in \argmin \left\{ \frac{q^\circ(X_k)}{p^\circ_{\psi}(X_k)}\,:\, \psi \in \widehat{\Psi}_k(\delta) \right\}. 
\label{eq:LUI}
\end{equation}
We have the following result.
\begin{proposition}
\label{prop:lui}
Suppose that $\widehat{\Psi}_k(\delta)$ is a $(1-\delta)$-confidence set for $\psi_k$ for all $\p \in \cP_k$. Then $E_k^{\mathrm{LUI}}$ defined in~\eqref{eq:LUI} is a $(0,\delta)$-approximate e-variable for $\cP_k$.
\end{proposition}
The construction of LUI e-values is analogous to the construction of p-values by maximizing over a confidence set for the nuisance parameter by~\citet{berger1994values}. Notice that if \smash{$\widehat{\psi}_k^{\mathrm{UI}} \in \widehat{\Psi}_k(\delta)$}, then $E_k^{\mathrm{LUI}}= E_k^{\mathrm{UI}}$. This simple observation allows us to highlight the difference of UI and LUI.
Under the alternative, $X_k \sim \mathbb Q_k$, it may hold (depending on $\mathbb Q_k$ and the construction of \smash{$\widehat{\Psi}_k(\delta)$}) that \smash{$\widehat{\psi}_k^{\mathrm{UI}} \notin \widehat{\Psi}_k(\delta)$} with high probability and thus also that $E_k^{\mathrm{UI}} < E_k^{\mathrm{LUI}}$.

Now suppose instead that \smash{$\widehat{\Psi}^K(\delta) \subseteq\Psi^K$} is a $(1-\delta)$-simultaneous confidence set for $\boldpsi = (\psi_1,\ldots,\psi_K)$ and let \smash{$\widehat{\Psi}_k(\delta)$} be the projection of \smash{$\widehat{\Psi}^K(\delta)$} onto the $k$-th coordinate. Then it follows that \smash{$ E_1^{\mathrm{LUI}},\ldots,  E_K^{\mathrm{LUI}}$} are $(0,\delta)$-approximate compound e-variables for $\cP_1,\ldots,\cP_K$. This construction, however, may be very conservative and does not borrow strength across problems to yield more powerful approximate compound e-variables.

There are two ideas that lead us to a more powerful construction. First, in compound decision problems, even though it may not be possible to learn individual nuisance parameters accurately (e.g., $\boldpsi = (\psi_1,\ldots,\psi_K)$), it may be possible to accurately learn the ensemble as represented by its empirical distribution
$$G(\boldpsi) = \frac{1}{K}\sum_{k \in \mathcal{K}} \delta_{\psi_k}.$$
Specifically, let $\mathcal{G}$ be the set of all distributions supported on $\Psi$. We say that $\widehat{\mathcal{G}}(\delta) \subseteq\mathcal{G}$ is a $(1-\delta)$-localization of $G(\boldpsi)$ if:
$$
\p( G(\boldpsi) \in   \widehat{\mathcal{G}}(\delta)    ) \geq 1-\delta \; \text{ for all }\; \p \in \cP.
$$
We provide an explicit example of such a construction in Section~\ref{subsec:compound_ttest}. The second idea is that we can profile over the nuisance parameters appearing in the optimal discovery compound e-values of Theorem~\ref{th:compound-optimal} by lifting to the space of distributions. Compound universal inference (CUI) e-values are defined as follows:\footnote{ 
\citet{nguyen2023finite} use universal inference for an empirical Bayes application. By contrast, here we use ideas from empirical Bayes and universal inference to provide a new construction of compound e-values.
}
\begin{equation}
E_k^{\mathrm{CUI}} = \frac{q^\circ(X_k)}{  \int p^\circ_{\psi}(X_k) \dd \widehat{G}_k^{\mathrm{CUI}}(\psi)},\,\,\, \; \widehat{G}_k^{\mathrm{CUI}}  \in \argmin \left\{ \frac{ q^\circ(X_k)}{ \int p^\circ_{\psi}(X_k) \dd G(\psi)}\,:\, G \in  \widehat{\mathcal{G}}(\delta) \right\}. 
\label{eq:CUI}
\end{equation}
\begin{proposition}
\label{prop:cui}
Suppose that $\widehat{G}(\delta)$ is a $(1-\delta)$-localization for $G(\boldpsi)$ for all $\p \in \cP$. Then $E_1^{\mathrm{CUI}},\ldots, E_1^{\mathrm{CUI}}$ defined in~\eqref{eq:CUI} are $(0,\delta)$-approximate compound e-variables for $\cP_1,\ldots,\cP_K$ under $\cP$.
\end{proposition}

\section{The simultaneous t-test sequence model}
\label{sec:ttest}

In this section we provide concrete examples for data-driven implementations of localized compound universal inference, optimal discovery compound e-variables, as well as of asymptotic e-variables and asymptotic compound e-variables under asymptotics in which either the number of hypotheses or the sample size grows (or both). We do so in the stylized but practically relevant setting of multiple testing with simultaneous t-tests. We refer to~\citet{smyth2004linear, westfall2004weighted, finos2007fdr, bourgon2010independent, lu2016variance, guo2017analysis,  ignatiadis2021covariate, hoff2022smaller, ignatiadis2024empirical} for further works that have used this setting to demonstrate the fundamental difference between single hypothesis testing and multiple testing.

In our setting, we observe $K$ groups of data of size $n$, that is, we observe $X=(X_k^1,\ldots,X_k^n)_{k \in \mathcal{K}}$. We consider distributions $\p$ according to which all observations are jointly independent and the data within each group are iid. The universe of distributions $\cP$ is understood to obey the above assumptions. Given any $\p \in \cP$, we write $\p_k$ for its $k$-th group-wise marginal so that $X \sim \p$ entails that 
\begin{equation}
\label{eq:iid_marginal}
X_k^i \simiid \p_k\;\; \text{ for }\;\; i=1,\dotsc,n.
\end{equation}
We seek to test the hypotheses,
\begin{equation}
\cP_k = \{ \p \in \cP\,:\,  \E^{\p_k}[X_k^i] = 0,\,\var^{\p_k}(X_k^i) \in (0,\infty)\}.
\label{eq:ttest_nulls}
\end{equation}
Throughout we write $\mu_k = \mu_k(\p) = \mathbb E^{\p_k}(X_k^i)$, $\sigma_k^2 = \sigma_k^2(\p)=\var^{\p_k}(X_k^i)$ and $\lambda_k = \mu_k/\sigma_k$. We will typically keep the dependence on $\p$ implicit. With this notation, we seek to test whether $\mu_k=0$ (equivalently, $\lambda_k=0$). 
To test the above hypotheses, we start by constructing summary statistics in a group-wise manner, namely, for each $k \in \mathcal{K}$,
\begin{equation}
\bar{X}_k^{(n)} = \frac{1}{n}\sum_{i=1}^n X^i_k,\;\;\;\, S^{(n)}_k = \sqrt{\frac{1}{n}\sum_{i=1}^n (X^i_k)^2},\;\;\;\, \hat{\sigma}_k^{(n)} = \sqrt{\frac{1}{n-1}\sum_{i=1}^n \left(X^i_k - \bar{X}^{(n)}_k \right)^2}.%\;\;\;\, %T_k = \frac{\sqrt{n} \bar{X}_k}{\hat{\sigma}_k}.
\label{eq:one_sample_summary_mtp_rewrite}
\end{equation}
We consider the testing task first under a normality assumption in Section~\ref{subsec:ttest_normal}, and then without normality in Section~\ref{subsec:ttest_beyond_normality}.

\subsection{Asymptotic compound e-values under normality}
\label{subsec:ttest_normal}

We first assume that all (null and non-null) marginals $\cP_k$ are normal, that is we further restrict the universe $\cP$ to $\cP^{\mathrm{N}}$, where
\begin{equation}
\label{eq:normal_universe}
\cP^{\mathrm{N}} =\{\p \in \cP\;:\; \p_k\, \text{ is normal and }\,  \ubar{\sigma}^2 \leq \var^{\p_k}(X_k^i) \leq \bar{\sigma}^2 \;\text{ for all }\; k\},
\end{equation}
for known constants satisfying $0<\ubar{\sigma} \leq \bar{\sigma} < \infty$. For $\p \in \cP^{\mathrm{N}}$,
the summary statistics in~\eqref{eq:one_sample_summary_mtp_rewrite} are sufficient. 
In what follows we treat $n \geq 3$ as fixed and consider asymptotics as $K \to \infty$; for this reason we omit explicit dependence on $n$, for instance, we write $\bar{X}_k$ instead of \smash{$\bar{X}_k^{(n)}$}.

Our setting is encompassed by the composite null setting with nuisance parameters that we described in Section~\ref{subsec:compound_ui}. The nuisance parameters under the null ($\psi_k$ in Section~\ref{subsec:compound_ui}) are given by the variances $\sigma_k^2$. Letting $\boldsigma^2 = (\sigma_1^2,\ldots,\sigma_K^2)$, a central object is the empirical distribution of the $\sigma_k^2$, that is,
\begin{equation}
G(\boldsigma^2) = \frac{1}{K}\sum_{k \in \mathcal{K}} \delta_{\sigma_k^2}.
\label{eq:G_sigma}
\end{equation}
Our goal is to mimic the optimal discovery compound e-values in~\eqref{eq:optimal_evalue} without knowledge of $G(\boldsigma^2)$.

\subsubsection{Localized compound universal inference for the simultaneous t-test}
\label{subsec:compound_ttest}

Our first approach is based on localized compound universal inference (CUI) of Section~\ref{subsec:compound_ui}. CUI requires a confidence set of distributions containing $G(\boldsigma^2)$. We can achieve this by working with the marginal distribution of the $\hat{\sigma}_k^2$ which is ancillary for the $\mu_k$.
Concretely, let
\begin{equation}
f_G(\hat{\sigma}_k^2) =  \int_0^{\infty} p(\hat{\sigma}_k^2 \mid \sigma^2)\dd G(\sigma^2),\;\;\;\,  p(\hat{\sigma}_k^2 \mid \sigma^2) = \frac{\nu^{\nu/2}\left(\hat{\sigma}_k^2\right)^{\nu/2-1}}{\left(\sigma^2\right)^{\nu/2} 2^{\nu/2}\Gamma(\nu/2)}   \exp\left(-\frac{\nu \hat{\sigma}_k^2}{2\sigma^2}\right),
\label{eq:marginal_density_defi}
\end{equation}
where $\nu = n-1$ denotes the degrees of freedom. Let $\mathcal{G}$ be the class of all distributions supported on $[\ubar{\sigma}^2, \bar{\sigma}^2]$ and consider:
\begin{equation}
\widehat{\mathcal{G}}(\delta) = \left\{G \in \mathcal{G}\,:\, \sup_{t \in (0,\infty)} \left|  \frac{1}{K} \sum_{k \in \mathcal{K}} \id_{ \{\hat{\sigma}_k^2 \leq t\}} - \int_0^t f_G(u) \dd u \right| \leq \sqrt{\frac{1+\log(2/\delta)}{2K}} \right\}. 
\label{eq:floc_ttest}
\end{equation}

\begin{proposition}
\label{prop:ttest_cui}
For any $\p \in \cP^{\mathrm{N}}$ it holds that:
$$
\p( G(\boldsigma^2) \in \widehat{\mathcal{G}}(\delta)) \geq 1-\delta.
$$
\end{proposition}
\begin{proof}
This follows by the Dvoretzky–Kiefer–Wolfowitz inequality with Massart's tight constant~\citep{massart1990tight} and Bretagnolle's argument~\citep{bretagnolle1981}; see also 
\citet[Lemma 7.1]{donoho2006asymptotic}. The additive factor of $1$ on the right-hand side is needed because $\hat{\sigma}_k^2$ are not identically distributed.
\end{proof}

The construction is akin to the $F$-Localization of~\citet{ignatiadis2022confidence} with the difference that we do not require that \smash{$\sigma_k^2 \simiid G$} and instead construct the localization in the compound setting.

\begin{remark}[Computation of CUI]
\label{remark:cui_computation}
Computing CUI amounts to solving the optimization problem in~\eqref{eq:CUI}. To do this in practice, we proceed as follows. First, following~\citet{ignatiadis2022confidence}, we discretize $\mathcal{G}$ by taking a covering of $[\ubar{\sigma}^2,\bar{\sigma}^2]$  with $B$ points $\tilde{\sigma}_1^2,\dots,\tilde{\sigma}_B^2 \in [\ubar{\sigma}^2,\bar{\sigma}^2]$ and considering all distributions of the form \smash{$G= \sum_{\ell=1}^B \pi_{\ell} \delta_{\tilde{\sigma}_{\ell}^2}$}, where $(\pi_1,\dots,\pi_B)$ lies on the probability simplex. Second, we relax the Kolmogorov-Smirnov constraint in~\eqref{eq:floc_ttest} by enforcing it only on a finite subset $\mathcal{T} \subseteq (0, \infty)$; e.g., we can let $\mathcal{T}$ be a subset of the realized sample variances $\hat{\sigma}_1^2,\dotsc, \hat{\sigma}_K^2$. This second discretization step is always conservative, that is, it makes the localization strictly larger. Our computational approximation to $E_k^{\mathrm{CUI}}$ is then given by $q^\circ(X_k)/ \mathrm{obj}$, where $\mathrm{obj}$ is the optimal value of the following linear program:
\begin{align*}
\underset{\pi_1,\ldots,\pi_B}{\text{maximize}} \quad & \sum_{\ell=1}^B \pi_{\ell} p^{\circ}_{\tilde{\sigma}_{\ell}^2}(X_k) \\
\text{subject to} \quad & \sum_{\ell=1}^B \pi_{\ell} = 1,\quad \pi_{\ell} \geq 0, \quad \ell = 1, \ldots, B, \\
& \left| \frac{1}{K} \sum_{j \in \mathcal{K}} \id_{\{\hat{\sigma}_j^2 \leq t\}} - \sum_{\ell=1}^B \pi_{\ell} \int_0^t p(u \mid \tilde{\sigma}_{\ell}^2) \, \dd u \right|  \leq \sqrt{\frac{1+\log(2/\delta)}{2K}}, \quad  t \in \mathcal{T}.
\end{align*}

\end{remark}

\subsubsection{Optimal discovery compound e-values via empirical Bayes}
\label{subsec:odp_ebayes}

In this section, we take a more direct approach to approximating the optimal discovery compound e-values. We directly use empirical Bayes to estimate the compound empirical distribution $G(\boldsigma^2)$ in~\eqref{eq:G_sigma}.

It will be convenient to parameterize the distribution of $X_k$ by $(\lambda_k, \sigma_k^2)$ where $\lambda_k = \mu_k / \sigma_k$, that is, for $x \in \R^n$, we write:
$$
p_{\lambda, \sigma^2}(x) = \frac{1}{(2\pi \sigma^2)^{n/2}} \prod_{i=1}^n \exp\left( - \frac{(x_i-\lambda \sigma)^2}{2 \sigma^2} \right),\;\;\; p_{\sigma^2}(x) = p_{0,\sigma^2}(x).
$$
In this way, given a measure $Q$ over $\R \times (0,\infty)$ and $G$ over $(0,\infty)$, we may define the Bayes factor
\begin{equation}
E(x; G, Q) = \frac{\int p_{\lambda, \sigma^2}(x)  \dd Q(\lambda, \sigma^2)}{\int p_{\sigma^2}(x) \dd G(\sigma^2)}.
\end{equation}
For fixed $Q$ and with $G=G(\boldsigma^2)$, it holds that 
$$E(x; G(\boldsigma^2), Q) =   \frac{K \int p_{\lambda, \sigma^2}(x)  \dd Q(\lambda, \sigma^2)}{\sum_{j \in \mathcal{K}} p_{\sigma_j^2}(x)},$$
and so, $E(X_k; G(\boldsigma^2), Q)$,  $k \in \mathcal{K}$, are optimal discovery compound e-variables. We propose to construct asymptotic compound e-variables 
\begin{equation}
\label{eq:EKEB}
E_k = E(X_k; \widehat{G}, \widehat{Q}),
\end{equation}
by plugging in specific estimates $\widehat{G}$ of $G$ and $\widehat{Q}$ of $Q$ as described below.
We estimate $\widehat{G}$ based on the nonparametric maximum likelihood estimator (NPMLE) of $G$ by~\citet{robbins1950generalization} and \citet{kiefer1956consistency} based on $\hat{\sigma}_1^2,\ldots,\hat{\sigma}_K^2$:
\begin{equation}
\label{eq:npmle_opt}
\widehat{G} \in \argmax\left\{ \sum_{k \in \mathcal{K}} \log\left( f_G(\hat{\sigma}_k^2)\right)\,:\,G \in \mathcal{G}\right\},\;\;
\end{equation}
where we recall that $\mathcal{G}$ is the class of all distributions supported on $[\ubar{\sigma}^2, \bar{\sigma}^2]$ and $f_G(\hat{\sigma}_k^2)$ is defined in~\eqref{eq:marginal_density_defi}.
Moreover we let \smash{$\widehat{Q} = \widehat{H} \otimes \widehat{G}$} where \smash{$\widehat{H}$} is a distribution on $\mathbb R$ estimated based on $X_1,\ldots,X_K$ and that is supported on $[\ubar{\lambda},\bar{\lambda}] \subseteq \R$, where $\ubar{\lambda},\bar{\lambda}$ are fixed.

\begin{theorem}
\label{theorem:eb_asymptotic}
Consider asymptotics with $n \geq 3$ fixed and $K \to \infty$. Assume
there exists $H^*$ such that \smash{$\widehat{H} \cd H^*$} almost surely. Then $E_1,\ldots,E_K$ in \eqref{eq:EKEB} are  strongly asymptotic compound e-variables for $\cP_1,\ldots,\cP_K$ under \smash{$\cP^{\mathrm{N}}$}.
\end{theorem}

Theorem~\ref{theorem:eb_asymptotic} is a purely frequentist result. As is common in compound decision theory, we ``pretend'' that $\sigma_k^2 \sim G$ when computing the NPMLE, however, in the theorem statements, $\sigma_1^2,\ldots,\sigma_K^2$ are fixed (deterministic). Moreover, $H^*$ can be any distribution and does not need to represent the data-generating distribution.

Since the ODP was first described by~\citet{storey2007optimal}, several authors have proposed data-driven implementations of the ODP when the nuisance parameters are unknown including~\citet{storey2007optimala, guindani2009bayesian, woo2011computationally}. The approach closest to ours is the one of~\citet{noma2012optimal, noma2013empirical} who also pursue an empirical Bayes approach; however, they only consider parametric specifications for $G$ (while allowing for nonparametric $H$). To the best of our knowledge, Theorem~\ref{theorem:eb_asymptotic} presents the first theoretical guarantee for data-driven ODP and together with Theorem~\ref{theo:eBH_controls_the_FDR} it implies that we can use the empirical Bayes optimal discovery compound e-values alongside the e-BH procedure to asymptotically control the FDR.

\begin{remark}[Computation of NPMLE]
\label{remark:computation_npmle}
In practice we compute the NPMLE~\eqref{eq:npmle_opt} using the standard strategy proposed by~\citet{koenker2014convex}: we first discretize $G$, as described in Remark~\ref{remark:cui_computation}, and this discretization reduces~\eqref{eq:npmle_opt} to a finite conic programming problem that can be solved  by interior point methods.
\end{remark}

We demonstrate the improved power of our two constructions through a simulation study in Section~\ref{subsec:ttest_simulation}.

\subsection{Asymptotic e-values and compound e-values beyond normality}
\label{subsec:ttest_beyond_normality}

In this section, we demonstrate how to construct asymptotic (compound) e-values without assuming normality. 

\subsubsection{A prototypical example of an asymptotic e-value}
\label{subsec:prototypical_single_asym}
We first provide a prototypical construction of asymptotic e-variables 
along the lines of 
the central limit theorem (as discussed in Section~\ref{sec:asymp-e}).  We consider the setting
described in the beginning of Section~\ref{sec:ttest}
and also use the summary statistics in~\eqref{eq:one_sample_summary_mtp_rewrite} (even though they are no longer sufficient in the absence of normal errors). We focus on the case of a single hypothesis, i.e., $K=1$ and consider asymptotics as $n \to \infty$. For this reason, we omit the subscript $k$, and write e.g., \smash{$\bar{X}^{(n)} = \bar{X}_1^{(n)}$}. We seek to test the hypothesis $\mathcal{P}_1$ in~\eqref{eq:ttest_nulls}.

For fixed $\lambda \in \R$, let
\begin{equation}
\label{eq:asymptotic_evalue}
E^{(n)} = \exp\left(\lambda \frac{ \sqrt{n} \bar{X}^{(n)}}{S^{(n)}} - \frac{\lambda^2}{2}\right),
\end{equation}
and
\begin{equation}
\label{eq:asymptotic_evalue2}
\tilde{E}^{(n)} = \frac12 \exp\left(\lambda \frac{\sqrt{n}\bar{X}^{(n)}}{S^{(n)}}- \frac{\lambda^2}{2}\right) + \frac12 \exp\left(-\lambda \frac{\sqrt{n}\bar{X}^{(n)}}{S^{(n)}} - \frac{\lambda^2}{2}\right).
\end{equation}
The next proposition justifies that \eqref{eq:asymptotic_evalue} and~\eqref{eq:asymptotic_evalue2} define sequences of strongly asymptotic e-variables.

\begin{theorem}
\label{prop:c9-asym}
The sequences of random variables $E^{(n)}$ in \eqref{eq:asymptotic_evalue}
and $\tilde{E}^{(n)}$ in \eqref{eq:asymptotic_evalue2} 
are strongly asymptotic e-variables for $\cP_1$. 
If $\hat{\sigma}^{(n)}$ is used in place of $S^{(n)}$ in \eqref{eq:asymptotic_evalue} and \eqref{eq:asymptotic_evalue2}, then $E^{(n)}$ and $\tilde{E}^{(n)}$ are asymptotic e-variables for $\cP_1$.
In either case, for $\lambda>0$, $E^{(n)}$ (resp.~$\tilde{E}^{(n)}$) grows to infinity with probability $1$ under the alternative that $X^1,\dots,X^n$ are drawn iid from $\mathbb Q$ with $\E^{\mathbb Q}[X^i]> 
0$ (resp.~$\E^{\mathbb Q}[X^i]\ne 
0$) and $\var^{\mathbb Q}(X^i) < \infty$. 
\end{theorem}

The above conclusion extends to a broader class of asymptotic e-variables obtained by mixing in~\eqref{eq:asymptotic_evalue} over $\lambda$ with respect to a distribution with bounded support. Concretely, let $G$ be any distribution on $\mathbb R$ with compact support. Then, under the assumptions of this section,
$$
\int \exp\left(\lambda \frac{\sqrt{n}\bar{X}^{(n)}}{S^{(n)}} - \frac{\lambda^2}{2}\right) \d G(\lambda),
$$
is a strongly asymptotic e-variable. 

One central purpose of defining asymptotic e-variables is to handle cases where one does not know how to construct a nonasymptotic e-variable. Indeed, for the above $\cP_1$, we do not know of nonasymptotic e-variables (other than constants).

\subsubsection{A prototypical construction of asymptotic compound e-values}
\label{subsec:sum_of_squares_compound}

We next describe a prototypical model of constructing asymptotic compound e-variables.
The construction appeared in~\citet{ignatiadis2024evalues} (assuming normality as in Section~\ref{subsec:ttest_normal}) and here we generalize the construction and show that it does not require normality. We start with the summary statistics in~\eqref{eq:one_sample_summary_mtp_rewrite}, and will consider asymptotics with $\max\{n,K\} \to \infty$.
\citet{ignatiadis2024evalues} were interested in computing approximate compound e-variables that satisfy \smash{$E_k^{(n)} \propto (S_k^{(n)})^2$} for all $k \in \mathcal{K}$. 
To motivate the construction that follows, fix any $\p \in \cP$ and write  $\mu_k(\p) =\E^{\p_k}[X_k^i]$ and $\sigma_k^2(\p) = \var^{\p_k}(X_k^i)$. 
Observe that for any $k \in \mathcal{K}$, 
$$\E^{\p_k}[(S_k^{(n)})^2] = \sigma_k^2(\p) + \mu_k^2(\p) \, \text{ for all }\,  k \in \mathcal{K} \quad\text{and so}\quad \E^{\p_k}[(S_k^{(n)})^2] = \sigma_k^2(\p)\, \text{ for }\, \p \in \cP_k.
$$
One can check that for any $\p \in (\bigcup_{k \in \mathcal{K}} \cP_k)\bigcap \cP$, 
$$
\sum_{k: \p \in \cP_k} \E^{\p}\left[\tilde{E}_k^{(n)}(\p)\right]  = K, \quad \text{where} \quad \tilde{E}_k^{(n)}(\p) = \frac{K \big(S_k^{(n)}\big)^2}{ \sum_{j: \p \in \cP_j} \sigma_j^2(\p)}.
$$
The above construction is not directly usable in practice because of the dependence of $\tilde{E}_k^{(n)}(\p)$ on $\p$. However, it motivates a construction in which we conservatively estimate the $\p$-dependent quantity $\sum_{j: \p \in \cP_j} \sigma_j^2(\p)$ by $\sum_{j \in \mathcal{K}}(\hat{\sigma}^{(n)}_{j})^2$ and then let
\begin{equation*}
E^{(n)}_k = \frac{K \big(S_k^{(n)}\big)^2}{\sum_{j \in \mathcal{K}}\big(\hat{\sigma}^{(n)}_{j}\big)^2},\quad k \in \mathcal{K}.
\end{equation*}
We will show that the above are asymptotic compound e-variables under an asymptotic setup indexed by $m \in \mathbb N$ such that $n=n(m)$, $K=K(m)$ and $\max\{n, K\} \to \infty$ as $m \to \infty$. Our formal result will be stated in terms of two universes of distributions. The first set imposes some restrictions on moments of the marginals,
$$
\cP^{\mathrm{B}} =\{\p \in \cP\;:\; \E^{\p_k}[|X_k^i|^{2(1+\delta)}] \leq C,\;\var^{\p_k}(X_k^i) \geq c \;\text{ for all }\; k\},
$$
where $\delta \in (0,1], c, C \in (0, \infty)$ are constants. The second set is equal to the universe $\mathcal{P}^{\mathrm{N}}$ in~\eqref{eq:normal_universe}.

\begin{proposition}
\label{proposition:simultaneous_ttest_variance_compound}
Suppose that $n=n(m) \geq 2$ and $K=K(m) \geq 1$ are such that $\max\{n, K\} \to \infty$ as $m \to \infty$. 
Then:
\begin{enumerate}[label=(\roman*)]
\item $(E_1^{(n)},\ldots, E_K^{(n)})_{m\in \N}$ is a sequence of asymptotic compound e-variables for $(\cP_1,\ldots,\cP_K)$ under $\cP^{\mathrm{B}}$. 
\item $(E_1^{(n)},\ldots, E_K^{(n)})_{m\in \N}$ is a sequence of  strongly asymptotic compound e-variables for $(\cP_1,\ldots,\cP_K)$ under $\cP^{\mathrm{N}}$. 
\end{enumerate}
\end{proposition}

\section{Simulation study}
\label{subsec:ttest_simulation}

The goal of this section is to empirically demonstrate that the approximate/asymptotic compound e-values of Section~\ref{subsec:ttest_normal} can lead to a substantial power increase for the e-BH procedure. We follow the notation and setting of Section~\ref{subsec:ttest_normal}.

We simulate from the simultaneous normal t-test model with $K=2000$ and let $n \in \{5,10\}$ be a simulation parameter. We consider two settings for \smash{$\sigma_k^2$}, either all \smash{$\sigma_k^2=1$} or \smash{$\sigma_k^2 \simiid \mathrm{Unif}[0.5,2]$}. Moreover we set $\lambda_k=0$ for $1800$ hypotheses (nulls), and $\lambda_k = \xi/\sqrt{n}$ for the other 200 hypotheses. We call $\xi$ the effect size and vary it in the interval $[2,6]$ in our simulations. 

We apply e-BH at level $\alpha=0.1$ and compare the following methods. We provide all methods with access to $\xi/\sqrt{n}$ and they all thus set $H = \delta_{\xi/\sqrt{n}}$. This is favorable for all methods but allows us to isolate the impact of the denominators in constructing (asymptotic, compound) e-values. 

\begin{enumerate}
\item z-Oracle: This is an idealized method that has access to the true $\sigma_k^2$ for all $k$ and sets $E_k^z = E(X_k; \delta_{\sigma_k^2}, H \otimes \delta_{\sigma_k^2})$.
\item EB-Oracle: This is another idealized method with access to the empirical distribution $G(\boldsigma^2)$. It is equal to the (oracle) optimal discovery compound e-values $E(X_k; G(\boldsigma^2), H \otimes G(\boldsigma^2))$.
\end{enumerate}
For the remaining methods, we use the universe in~\eqref {eq:normal_universe} with \smash{$\ubar{\sigma}^2 = 10^{-3}$} and \smash{$\bar{\sigma}^2= 10^{3}$}. 
The remaining methods are data-driven (i.e., are not given any extra information about the nuisance parameters beyond the wide bounds $\sigma_k^2 \in [\ubar{\sigma}^2, \bar{\sigma}^2]$). For the discretization of $G$ (described in Remarks~\ref{remark:cui_computation} and~\ref{remark:computation_npmle}), we take $B=600$ and consider a logarithmically equispaced grid from $\ubar{\sigma}^2$ to $\bar{\sigma}^2$. Below, \smash{$\widehat{G}$} refers to the NPMLE.
\begin{enumerate}[resume]
\item t-test: This is the scale invariant e-value studied in depth by~\citet{wang2025anytimevalid}. The $k$-th e-value is a measurable function of the t-statistic $T_k = \sqrt{n}\bar{X}_k / \hat{\sigma}_k$.
\item Empirical Bayes (EB): The plug-in compound e-values studied in Theorem~\ref{theorem:eb_asymptotic}, that is, $E_k^{\text{EB}} = E(X_k;\widehat{G},  H \otimes \widehat{G})$.
\item Universal inference (UI):  As in~\eqref{eq:UI} with numerator given by $\int p_{\lambda, \sigma^2}(x)  \dd (H\otimes \widehat{G})(\lambda, \sigma^2)$.
\item Compound universal inference (CUI): as in Section~\ref{subsec:compound_ttest} with the same numerator as for UI and the localization in~\eqref{eq:floc_ttest} with $\delta=0.01$. We apply e-BH at level $0.09$ rather than $0.1$.\footnote{Formally, UI and CUI are not fully justified as they estimate the numerator based on the data. However, we expect them to be more conservative than EB, and so include them as a comparison of different treatments for the denominator.}
\end{enumerate}
We evaluate results in terms of FDR and Power, which we define as the expected proportion of true discoveries among non-null hypotheses. These metrics are estimated by averaging the results over 200 Monte Carlo replications of each simulation.

\begin{figure}[t] 
\centering 
\includegraphics[width=\linewidth]{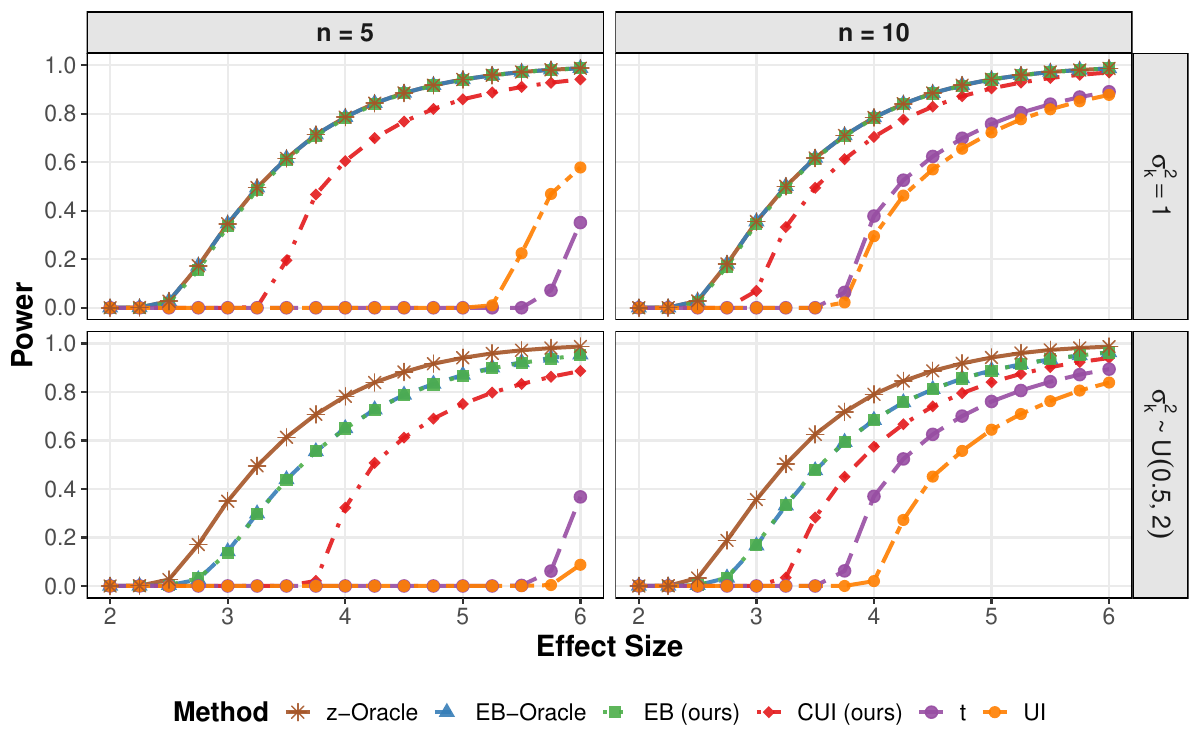} \caption{Power comparison of six multiple testing procedures under different sample sizes ($n \in \{5, 10\}$) and variance settings ($\sigma_k^2 = 1$ vs. $\sigma_k^2 \sim \mathrm{Unif}[0.5, 2]$). Methods include z-Oracle, EB-Oracle, EB (Empirical Bayes), CUI (Compound Universal Inference), t-test, and UI (Universal Inference). Among data-driven methods (i.e., methods that do not have oracle knowledge regarding the variances), our proposed EB method has by far the largest power. Moreover, EB is indistinguishable from its oracle counterpart (EB-Oracle). 
} \label{fig:simulation_results} 
\end{figure}

The results of the simulation are shown in Figure~\ref{fig:simulation_results}. All methods control FDR at the nominal $\alpha=0.1$ (not shown). In all cases, the most powerful methods is z-Oracle, which has the most information. When all $\sigma_k^2=1$, then z-Oracle and EB-Oracle are identical (since knowledge that $G(\boldsigma^2)=\delta_{1}$ implies that $\sigma_k^2=1$ for all $k$). Meanwhile, our data-driven version (EB), is effectively indistinguishable from EB-Oracle, showing that the NPMLE effectively learns $G(\boldsigma^2)$. CUI is less powerful than EB, but has more power than standard e-value methods (UI and t). The relative performance of UI and t depends on the simulation, see~\citet{wang2025anytimevalid} for further discussion. Overall, this simulation demonstrates the  power benefits by using asymptotic compound e-values instead of (bona fide) e-values.

\section{Summary}
We offer a formal treatment of asymptotic and compound e-values by giving them rigorous definitions. We provide several prototypical constructions, and explain why compound e-values are a fundamental concept in multiple testing. We demonstrate several ways of constructing bona fide, approximate, or asymptotic compound e-values, for example,
\begin{enumerate}
    \item by leveraging any FDR controlling procedure;
    \item through the optimal discovery compound e-values (which are ratios of mixture likelihoods) and empirical Bayes approximations thereof;
    \item through localized compound universal inference;
    \item or by calibration of (approximate) compound p-values.
\end{enumerate} 
Since notions related to compound e-values have already been extensively used in the recent multiple testing literature, we believe that our formalism will help streamline the development of new procedures, and provide a common language for future research in multiple testing that builds upon compound e-values.

\paragraph{Reproducibility.} We provide code to reproduce the numerical results of this paper on Github:\\
\url{https://github.com/nignatiadis/compound-evalues-paper}

\paragraph{Acknowledgments.} This work was completed in part with resources provided by the University of Chicago’s Research Computing Center.
RW is supported by the Natural Sciences and Engineering Research Council of Canada (CRC-2022-00141, RGPIN-2024-03728). NI gratefully acknowledges support from the U.S. National Science Foundation (NSF DMS-2443410).

\bibliographystyle{abbrvnat}
\bibliography{references}

\appendix 

\section{Compound p-value definition in~\citet{habiger2014compound}}
\label{sec:habiger_compound}
We briefly point out that the term compound p-values has previously been used in a way that is distinct from Definition~\ref{defi:compound_evalues}.~\citet{habiger2014compound} considered a setting akin to the sequence model in Section~\ref{sec:sequence_model} wherein $X=(X_k: k \in \mathcal{K})$ and in principle $X_k$ can be used to test $\cP_k$ for all $k \in \mathcal{K}$. Let $P_1,\dots,P_K$ be $[0,1]$-valued random variables. \citet{habiger2014compound} called these ``compound p-values'' if they are bona fide p-values (that is $\p(P_k \le t) \le t$ for all $t \in [0,1]$ and all $\p \in \mathcal{P}_k$) and depend also on data for other hypotheses $(X_{\ell}: \ell \neq k)$, i.e., $P_k = P_k(X)$ rather than $P_k = P_k(X_k)$. In other words, ``compound p-values'' are not separable as per Definition~\ref{defi:simple_separable}, and so could be called ``non-separable p-values'' instead. The definition treats separable p-values as the default data analysis choice (which is reasonable given how p-values are often used in practice), and one needs a qualifier to refer to non-separable p-values. By contrast, we prefer to think of (compound) p-values and e-values as being non-separable by default (given their flexibility and generality), and to use an explicit qualifier for the separable case. Moreover, we use the term compound in relation to the averaging (compounding) operation inherent in Definition~\ref{defi:compound_evalues}.

\section{More on *approximate compound p-values and e-values}
\label{sec:alt_defi}

Here we elaborate further on Definition~\ref{defi:additive_approx_compound}.
In analogy to Definition \ref{defi:asymp_compound_evariables_pvariables}, we define (uniformly) *asymptotic compound p-variables and e-variables, requiring, e.g., that $\lim_{m\to \infty}  \varepsilon_m(\p) = \lim_{m\to \infty} \delta_m(\p) = 0$ for each $\p \in \cP$. We do not define strongly *asymptotic compound p-variables (or e-variables) since the natural definition is equivalent (under the atomless property) to strongly asymptotic compound p-variables (or e-variables).

Proposition~\ref{proposition:approximate_evalues_to_approximate_compound} has the following analogous result (whose proof we omit).
\begin{proposition}
Suppose that for all $k \in \mathcal{K}$, $P_k$ (resp.~$E_k$) is an $(\varepsilon_k,\delta_k)$-approximate p-variable (resp.~e-variable) for $\cP_k$. Then, $P_1,\dots,P_K$ (resp. $E_1,\dots,E_K$) are $(\varepsilon,\delta)^*$-approximate compound p-variables (resp. e-variables) for the choice:
$$
\varepsilon(\p) = \frac{1}{K} \sum_{k: \p \in \cP_k} \varepsilon_k(\p),\qquad \delta(\p) = \frac{1}{K}\sum_{k: \p \in \cP_k} \delta_k(\p).
$$
\end{proposition}
We can calibrate $*$approximate compound p-variables to $*$approximate compound e-variables.
\begin{proposition}
\label{prop:calibration_additively_approximate_compound}
 Let $\varepsilon:\cP \to \R_+ $ and $\delta:\cP \to [0,1]$. 
    Calibrating $(\varepsilon, \delta)^*$-approximate  compound p-variables  yields  $(\varepsilon, \delta)^*$-approximate compound e-variables, and vice versa.
\end{proposition}

\begin{proof}
We only prove the calibration result from *approximate compound p-variables to e-variables. To this end, let $P_1,\dots,P_K$ be $(\varepsilon, \delta)^*$- approximate compound p-variables for $\cP_1,\dots,\cP_K$ and let $h$ be an e-to-p calibrator. Fix $\p \in \cP$.
Define $P_L$ and $\tilde{P}$ as in the proof of Theorem~\ref{prop:compound-pe1} and $\widetilde{\cP}=\{\p\}$ as a singleton null hypothesis. Then,  $\tilde{P}$ is an $(\varepsilon, \delta)$-approximate p-variable for $\widetilde{\cP}$. Thus, by Theorem~\ref{th:calibration}, 
$h(\tilde{P})$ is an $(\varepsilon, \delta)$-approximate e-variable for $\widetilde{\cP}$. 

Hence, further arguing as in the proof of Theorem~\ref{prop:compound-pe1}, for any $t>0$,
$$ \frac{1}{K} \sum_{k: \p \in \mathcal{P}_k }\E^{\p}\left[   h( P_k) \wedge t \right] = \frac{K_0}{K} 
 \E^{\p}\left[   h( P_{L} \wedge t ) \right] 
 \le   \E^{\p}\left[ h(\tilde P)\wedge t \right] \leq 1+\varepsilon_\p  +\delta_\p t.$$
Thus $h(P_1),\dots,h(P_K)$ are $(\varepsilon,\delta)^*$-approximate e-variables for $\cP_1,\dots,\cP_K$.
\end{proof}

Theorem~\ref{theo:eBH_controls_the_FDR} demonstrated that approximate compound e-variables are naturally compatible with approximate FDR control when using the e-BH procedure. An analogous statement is not true for $^*$approximate compound e-variables without further conditions. However it is true in a regime where e-BH makes sufficiently many discoveries. The statement of the following proposition is inspired by~\citet[Theorem 2]{li2019multiple}, whose FDR bound also depends on the tail probability of the number of rejections.

\begin{proposition}
\label{prop:eBH_star_approximate}
Let $\cD$ be the e-BH procedure at level $\alpha \in (0,1)$ and 
suppose that $E_1,\ldots,E_K$ are $(\varepsilon,\delta)^*$-approximate e-variables for $(\cP_1,\ldots,\cP_K)$ under $\cP$.  Then
$$
\mathrm{FDR}_{\cD}^{\p} \leq \alpha(1+\varepsilon_{\p}) + \inf_{\kappa \in (0,1)}  \left \{\frac{\delta_{\p} }{\kappa}  + \p(R_{\cD} < K \kappa) \right\}
\leq \alpha(1+\varepsilon_{\p}) +     \sqrt{\delta_{\p}}    + \p(R_{\cD} < K \sqrt{\delta_{\p}})  .
$$
Under an asymptotic regime with $\delta_{\p}, \varepsilon_{\p} \to 0$ and $\p(R_{\cD} < K \kappa) \to 0$ for fixed $\kappa \in (0,1)$,  then e-BH asymptotically controls the $\mathrm{FDR}$ at level $\alpha$.
\end{proposition}

\begin{proof}
Consider the event $A = \{R_{\cD} \geq K \kappa\}$. On this event we have that $K/(\alpha R_{\cD}) \leq 1/(\alpha \kappa)$. 
$$
\begin{aligned}
\mathbb E^{\p}\left[\frac{V_{\cD}}{R_{\cD}} \id_A \right] &= \sum_{k: \p \in \cP_k } \mathbb E^{\p}\left[\frac{\id_{\{E_k \geq K/(\alpha R_{\cD})\}}}{R_{\cD}}\id_A \right] \\ 
&\leq \sum_{k: \p \in \cP_k } \mathbb E^{\p}\left[\frac{\id_{\{E_k \land (1/(\alpha\kappa)) \geq K/(\alpha R_{\cD})\}}}{R_{\cD}} \right] \\ 
&\leq \frac{\alpha}{K} \sum_{k: \p \in \cP_k } \mathbb E^{\p}\left[ E_k \land (1/(\alpha\kappa)) \right] \\ 
& \leq \alpha \left( 1+\varepsilon_{\p} + \delta_{\p}\frac{1}{\alpha \kappa}\right).
\end{aligned}
$$
The last inequality holds by definition of $(\varepsilon,\delta)^*$-approximate e-variables. The other two results follow immediately.
\end{proof}

\section{Compound e-values in multiple testing: addendum}
\label{sec:compound_mtp_addendum}

\subsection{Combination and derandomization}
\label{subsec:derandomization}
The connections between FDR controlling procedures, e-BH, and compound e-values laid out in Section~\ref{sec:multiple_testing} motivate the following general and practical mechanism for combining discoveries across multiple testing procedures. To be concrete, suppose we run $L$ different multiple testing procedures $\mathcal{D}_1,\dots,\mathcal{D}_L$. The $\ell$-th procedure is applied to a subset of hypotheses $\{\cP_k:\; k \in \mathcal \mathcal{S}_{\ell}\}$ where $\mathcal{S}_{\ell} \subseteq \mathcal{K}$ and controls the FDR at level $\alpha_{\ell}$.  Then we may proceed as follows:
\begin{enumerate}
\item For the $\ell$-th multiple testing procedure, form compound e-values $E_k^{(\ell)},\, k \in \mathcal{S}_{\ell}$, for $\mathcal{P}_k,\,k \in \mathcal{S}_{\ell}$ (which always exist by the proof of Theorem~\ref{theo:universality_eBH}). 
They could be---but need not necessarily be---the implied ones formed in~\eqref{eq:universal-e} . For $k \in \mathcal{K} \setminus \mathcal{S}_{\ell}$, set $E_{k}^{(\ell)}=1$.
\item Fix weights $w_1,\dots,w_L \geq 0$ with $\sum_{\ell=1}^L w_{\ell} = 1$. Then construct new compound e-values $E_1,\dots, E_K$ by convex combination as in Example~\ref{exam:convex_combi}, i.e.,  $E_{k} = \sum_{\ell=1}^L w_{\ell} E_k^{(\ell)}$ for all $k \in \mathcal{K}$.
We also allow $w_1,\dots,w_L\ge 0 $ to be random, as long as they are independent of all e-values used in the procedures, and in that case it suffice to require $\sum_{\ell=1}^L \E^{\p} [w_{\ell}] \le  1$ for all $\p \in \cP$ (this condition is similar to the condition for compound e-values). 
\item Apply the e-BH procedure to the new compound e-values $E_1,\dots, E_K$ at level $\alpha$.
\end{enumerate}
The above construction is guaranteed to control the FDR at level $\alpha$, as long as all the individual procedures $\mathcal{D}_{\ell}$ form valid compound e-values.
Note that the value of $\alpha_{\ell}$ does not matter, as it is only used in the construction of $E_1^{(\ell)},\dots,E_K^{(\ell)}$ , possibly implicitly (e.g., in the proof Theorem~\ref{theo:universality_eBH}).   

One important application of the above recipe is derandomization. Suppose that $\mathcal{D}_{\ell}$ is a randomized multiple testing procedure, that is, it is a function of both the data $X$ as well as a random variable $U_{\ell}$ generated during the analysis. Such randomness may not be desirable, since different random number generation seeds will lead to different sets of discoveries. In such cases, the above recipe can be used to construct a new derandomized procedure $\mathcal{D}$ that is less sensitive to $U_1,\dots,U_L$ (formally, full derandomization would occur if $U_{\ell}$ are iid and $L\to \infty$). In this way~\citet{ren2024derandomised} were able to derandomize the model-X knockoff filter~\citep{ candes2018panning}, a flexible set of methods for variable selection in regression with finite-sample FDR control
 that previously relied on additional randomness.

A further application includes the following: \citet{banerjee2023harnessing} used the above recipe for meta-analysis in which the $\ell$-th study only reports the set of tested hypotheses, the set of discoveries, as well as the targeted FDR level. 

We end this subsection by using our recipe to derandomize the randomized e-BH procedures (including $R_1$-e-BH, $R_2$-e-BH, $R$-e-BH, and $U$-e-BH) proposed by~\citet{xu2024more}. For instance, $U$-e-BH draws $U \sim \mathrm{Unif}[0,1]$ and then applies e-BH to $E_1/U,\ldots,E_K/U$. Denote any of these procedures by $\mathcal{D}$. 
All of these procedures
have in common the following properties: they operate on e-values $E_1,\ldots,E_K$ as well as an external source of randomization $U$, they control the FDR without requiring further additional assumptions over e-BH, and finally they make at least as many discoveries as e-BH (that is, $\mathcal{D} \supseteq \mathcal{D}_{\text{e-BH}}$ almost surely), often making strictly more discoveries. Both~\citet{xu2024more} and~\citet{lee2024boosting} have noted the challenge of derandomizing these procedures; naive attempts lead back to e-BH applied to the original e-values. Now consider applying the derandomization procedure of this section with $\mathcal{D}_{\ell}$ denoting an application of $\mathcal{D}$ with external randomization seed $U_\ell$, $S_{\ell} = \mathcal{K}$, and weights $w_{\ell}=1$. Then the resulting procedure controls the FDR and is different from e-BH.

\subsection{Compound e-values as weights in p-value based multiple testing}
\label{subsec:compound_evalues_as_weights}

\cite{ignatiadis2024evalues} explained that e-values can be used as weights in p-value based multiple testing.
We explain why \emph{compound} e-values are the natural notion of weights for multiple testing with the p-BH procedure. We first define the p-BH procedure of~\citet{benjamini1995controlling} for FDR control based on p-values, as well as weighted generalizations thereof.

\begin{definition}[p-BH, weighted p-BH, ep-BH]
\label{defi:weighted-pBH}
Let $P_1,\dots, P_K$ be p-variables for the hypotheses $\cP_1,\dots,\cP_K$ and let $W_1,\dots, W_K$ be $[0,\infty]$-valued. Define $Q_k = P_k / W_k$ with the convention $0/0=0$. 
For $k\in \mathcal K$, let $Q_{(k)}$ be the $k$-th order statistic of $Q_1,\ldots,Q_K$, from the smallest to the largest. Consider the procedure that rejects hypothesis $k$ if $Q_k$ is among the smallest $k_q^*$ values $Q_1,\dots,Q_K$, where 
\begin{equation*} 
k_q^*:=\max\left\{k\in \mathcal K: \frac{K Q_{(k)}}{k} \le \alpha \right\},
\end{equation*}     
with the convention $\max(\varnothing) = 0$. 
When $W_k=1$ for all $k$, then the above procedure is the p-BH procedure of~\citet{benjamini1995controlling}. When $W_k=w_k$ is deterministic and such that $\sum_{k \in \mathcal{K}} w_k = K$, then the above procedure is the weighted p-BH procedure of~\citet{genovese2006false}. When $W_k=E_k$ and $E_1,\dots,E_K$ are compound e-variables, then the above procedure is the ep-BH (e-weighted p-BH) procedure of~\citet{ignatiadis2024evalues}.
\end{definition} 
The generalization of p-BH to allow for deterministic weights by~\citet{genovese2006false} was motivated by differentially prioritizing hypotheses to improve the power of p-BH while maintaining frequentist FDR control. The requirement $\sum_{k \in \mathcal{K}} w_k = K$ may be interpreted as a fixed size budget to be split across hypotheses. 
When the p-values are positive regression dependent on a subset (PRDS, \citet{benjamini2001control}) under the distributions in $\cP$ and the weights are deterministic then existing arguments in the literature~\citep{ramdas2019unified} demonstrate that the FDR of the weighted p-BH procedure is controlled at $(\alpha/K)\sum_{k \in \mathcal{N}} w_k$ under this assumption, which is bounded by $\alpha$ under the budget constraint $\sum_{k \in \mathcal{K}} w_k = K$. If the weights $W_1,\dots,W_K$ are random and independent of $P_1,\dots,P_K$, then by applying the above argument conditional on the weights and by iterated expectation, it follows that the FDR is controlled at $\sup_{\p\in \cP} (\alpha/K)\mathbb \sum_{k: \p \in \cP_k} \mathbb E^{\p}[W_k]$. The latter is controlled at $\alpha$ precisely when $W_1,\dots,W_K$ are compound e-values. We note that~\citet{ignatiadis2024evalues} defined the ep-BH procedure  and proved FDR control (in their Theorem 4) when $E_1,\dots,E_K$ are e-values. However,~\citet[Section 4.5]{ignatiadis2024evalues} explained that the guarantees would continue to hold for compound e-values (although their definition is slightly different from ours).

\subsection{Compound e-values as weights in merging p-values}
\label{sec:compound_evalues_to_merge}

A vector of arbitrarily dependent p-variables can be merged into one p-variables through p-merging functions in the sense of \cite{vovk2020combining}. 
All admissible homogeneous p-merging functions can be written as  
\begin{align}
\label{eq:def-f}
F(\mathbf p)  =  \inf\left\{\alpha \in (0,1):   \sum_{k=1}^K \lambda_k f_k\left(\frac{p_k}{\alpha}\right) \ge 1 \right\}  \mbox{ ~~~for $\mathbf p=(p_1,\dots,p_K)\in [0,\infty)^K$,}
\end{align}
for some  $\lambda_1,\dots,\lambda_K$ that are nonnegative weights summing to $1$ and  admissible p-to-e calibrators $f_1,\dots,f_K$ (\citet[Theorem 5.1]{vovk2022admissible}). Here, we set $\inf \emptyset=1$.

 Consider the global null setting $\mathbb P \in \bigcap_{k\in \mathcal K} \cP_k$, and let $E_1,\dots,E_K$ be compound e-variables for $\cP_1,\dots,\cP_K$.  Write $\mathbf E=(E_1,\dots,E_K)$.

 We follow the setting of \cite{ignatiadis2024evalues}, where one has access to a vector of p-values for the hypotheses $H_1,\dots,H_K$ that is independent of $\mathbf E$. This may be relevant in the context of follow-up experiments. 
 Suppose that $\mathbf P$ is a vector of p-variables for $\mathbb P$ independent of $\mathbf E$.
 We can define the e-weighted version of \eqref{eq:def-f} as
 \begin{align}
\label{eq:def-f2}
F(\mathbf p,\mathbf E)  =  \inf\left\{\alpha \in (0,1): \frac 1K  \sum_{k=1}^K E_k f_k\left(\frac{p_k}{\alpha}\right) \ge 1 \right\}  \mbox{ ~~~for $\mathbf p=(p_1,\dots,p_K)\in [0,\infty)^K$.}
\end{align}
The formulation \eqref{eq:def-f2}  includes \eqref{eq:def-f} as a special case via $E_k=\lambda_k K$ for $k\in \mathcal K$.
Note that for any calibrators $f_1,\dots,f_K$, the function $F(\mathbf p, \mathbf e)$ is decreasing in $\mathbf e$. This is intuitive, as larger e-values carry evidence against the null hypothesis, in the same direction as smaller p-values. 

For a clear comparison with some classic p-merging functions, in what follows, we allow  $F$ in \eqref{eq:def-f2} to be defined with
 any decreasing functions $f_1,\dots,f_K$  
satisfying $\int_0^1 f_k(x) \dd x \le 1$,  $f_k\ge 0$ on $[0,1]$, and $f_k\le 0$ on $(1,\infty)$ for each $k\in \mathcal K$. These functions are more general than p-to-e calibrators, and they generate p-merging functions  dominated by those generated by  calibrators.

  \begin{proposition}
In the above setting, we have 
 $$
 \mathbb P (F(\mathbf P,\mathbf E) \le \alpha ) \le \alpha~~~\mbox{for all $\alpha \in (0,1)$}.
 $$
 \end{proposition} 

 \begin{proof}
 Let $\beta \in (0,1)$.
 Note that $f_k^+:=f_k \vee 0 $ is a calibrator.
 For any calibrator $f$ and p-variable $P$, 
      we can check that $  \beta ^{-1}   f({P }/{\beta})  $  is an e-variable for $\mathbb P$; see e.g., Lemma 2.3 of \cite{gasparin2024combining}.
      Therefore, 
$$ \E^{\mathbb P} \left[\frac{1}{ K}\sum_{k=1}^K E_k f_k^+\left(\frac{P_k}{\beta}\right)\right]
\le \frac{1}{   K} \sum_{k=1}^K \E^{\mathbb P} [E_k]  \beta  
\le \beta .
$$
Markov's inequality gives 
$$
\mathbb P\left (\frac{1}{ K}\sum_{k=1}^K E_k f_k\left(\frac{P_k}{\beta}\right)   \ge 1 \right)\le \mathbb P\left (\frac{1}{ K}\sum_{k=1}^K E_k f_k^+\left(\frac{P_k}{\beta}\right)   \ge 1 \right) \le \beta.
$$
Using the monotonicity of $\beta \mapsto f_k(p_k/\beta)$,
we have, for any $\alpha\in (0,1)$,
  $$ F(\mathbf P,\mathbf E) \le \alpha \iff  \frac 1{K}  \sum_{k=1}^K E_k f_k\left(\frac{p_k}{\beta }\right) \ge 1  \mbox{~for all $\beta>\alpha$}.$$
  Therefore,
\begin{align*}
  \mathbb  P( F(\mathbf P,\mathbf E)  \le \alpha ) 
 & = \mathbb P \left( \bigcap_{\beta > \alpha}\left\{ \frac 1K  \sum_{k=1}^K E_k f_k\left(\frac{p_k}{\beta }\right) \ge 1 \right\} \right)\\ & \le \inf _{\beta > \alpha} \mathbb P \left(   \frac 1K  \sum_{k=1}^K E_k f_k\left(\frac{p_k}{\beta }\right) \ge 1  \right)
\le  \inf _{\beta > \alpha}  \beta =\alpha,
  \end{align*}
  showing the desired result. 
  \end{proof}

Some examples are discussed below. We will see that negative values for $f$ on $(1,\infty)$ allow for recovering some classic p-merging functions.
\begin{enumerate}
    \item[(i)]  Take $f_1,\dots,f_K$ as $f(x)=2-2x$. 
    We have $$F(\mathbf p,\mathbf E)=\inf\left\{\alpha \in (0,1): \frac{2}{K}\sum_{k=1}^K E_k(1-P_k/\alpha) \ge 1 \right\} = \left(\frac{\sum_{k=1}^K E_k P_k }{(\sum_{k=1}^K  E_k - K/2)_+}\right)\wedge 1, $$
    where we can set $F(\mathbf p,\mathbf E)=1$ if the denominator is $0$.
This can be seen as a e-weighted version of ``twice the mean" (truncated at $1$), which is precisely the case if $E_1=\dots=E_K=1$.
    \item[(ii)] Take $f_1,\dots,f_K$ as $f(x)=-\log x$.
     By writing
    $E=\sum_{k=1}^K E_k/K$, we have
    $$F(\mathbf p,\mathbf E)=\inf\left\{\alpha \in (0,1): \frac{-1}{K}\sum_{k=1}^K E_k \log (P_k/\alpha) \ge 1 \right\} = \left(\exp(E)\left( \prod_{k=1}^K P_k^{E_k} \right)^{1/(KE)}\right)\wedge 1. $$
This can be seen as a e-weighted version of ``$\exp(1)$ times the geometric mean"(truncated at $1$), which is precisely the case if $E_1=\dots=E_K=1$. 
 \end{enumerate}
 There is also a similar e-weighted version of the harmonic mean, but its formula is  more involved, and we  omit it here.

\subsection{Compound p-values in multiple testing}
\label{subsec:compound_pvalues_multiple_testing}
So far in this section we have focused on the role of compound e-values in multiple testing. Here we briefly discuss some known results about compound p-values. \citet{armstrong2022false} provided a counterexample of jointly independent compound p-values $P_1,\dots,P_K$, such that the p-BH procedure (Definition~\ref{defi:weighted-pBH}) does not control the FDR.
On the other hand, we show the following result, slightly generalizing a claim by~\cite{armstrong2022false}.

\begin{proposition}
    If $P_1,\dots,P_K$ are (asymptotic) compound p-values, then the p-BY~\citep{benjamini2001control} procedure (that is, the p-BH procedure applied at level $\alpha'=\alpha /  \sum_{k \in \mathcal{K}}1/k$) provides (asymptotic) control of the FDR at $\alpha$ under arbitrary dependence among $P_1,\dots,P_K$.
\end{proposition}
\begin{proof}
The nonasymptotic result was proved by~\cite{armstrong2022false}, but here is an alternate proof that makes it transparent to extend to the asymptotic case.
\cite{wang2022false} proved that the p-BY procedure may be equivalently described by two steps. First, the p-values $P_k$ are calibrated to e-values $E_k$ via a specific choice of calibrator $f:p\mapsto T(\alpha /(\ell_K p))/\alpha$, where
$$
T(x)  = \frac{K}{\lceil  K/x\rceil}\id_{\{x\ge 1\}}\mbox{~~with $T(\infty)=K$}.
$$
Second, the e-BH procedure is applied to $E_1,\dots,E_K$.   
With this, the nonasymptotic claim then follows directly from Theorems~\ref{prop:compound-pe1} and~\ref{theo:eBH_controls_the_FDR}. The asymptotic claim follows because calibrating asymptotic p-values yields asymptotic e-values, which when used with the e-BH procedure provides asymptotic FDR control.
\end{proof}

Furthermore,~\cite{armstrong2022false} proved that the p-BH procedure with compound p-values controls the FDR under weak dependence asymptotics (as in~\citet{storey2004strong}); such a property of the p-BH procedure (with compound p-values) was also anticipated earlier by~\citet{efron2007size}.

\section{Optimal simple separable compound e-values for general utility functions}
\label{sec:optimal_compound_general_utilities}

Here we present a result akin to Theorem~\ref{th:compound-optimal} for more general utility functions $U(\cdot)$ going beyond $U(\cdot)=\log(s)$.  Our notation follows that of Section~\ref{subsec:best_simple_separable}. Fix a utility function $U: [0, \infty] \to [-\infty,\infty]$.  We seek to solve the following optimization problem.
\begin{equation}
\label{eq:optim_compound_general}
    \begin{aligned}
    & \underset{s(\cdot)}{\text{maximize}} 
    & & \frac{1}{K}\sum_{k \in \mathcal{K}} \E^{\mathbb Q_k^{\circ}}[U(s(X_k))] \\
    & \text{subject to}
    & & s: \mathcal{X} \to [0,\infty] \\
    & & & s(X_1),\dots,s(X_K) \text{ are } \text{compound e-values}.
    \end{aligned}
\end{equation}
Note that optimization problem~\eqref{eq:optim_compound} is a special case of optimization problem~\eqref{eq:optim_compound_general} for the choice $U(\cdot) = \log(\cdot)$. The proof of the next theorem is analogous to that of Theorem~\ref{th:compound-optimal} and thus omitted.
\begin{theorem}[Optimal simple separable compound e-values; general utility]
\label{th:compound-optimal-general}
Let $\overline{\p}^{\circ} = \sum_{k \in \mathcal{K}} \p_k^{\circ}/K$ and $\overline{\mathbb Q}^{\circ} = \sum_{k \in \mathcal{K}} \mathbb Q_k^{\circ}/K$.
Suppose that $s^*(\cdot)$ solves the following optimization problem:
\begin{equation}
    \begin{aligned}
    & \underset{s(\cdot)}{\textnormal{maximize}} 
    & & \E^{X' \sim \overline{\mathbb Q}^{\circ}}[\log(s(X'))] \\
    & \textnormal{subject to}
    & & s : \mathcal{X} \to [0,\infty] \\
    & & & \E^{X' \sim \overline{\p}^{\circ}}[s(X')] \leq 1.
\end{aligned}
\label{eq:simple_opti}
\end{equation}
Then $s^*(\cdot)$ also solves optimization problem~\eqref{eq:optim_compound_general} and thus yields the optimal simple separable compound e-values for the utility function $U(\cdot)$.
\end{theorem}

By Theorem~\ref{th:compound-optimal-general} we can use existing results for problem~\eqref{eq:simple_opti} to solve problem~\eqref{eq:optim_compound_general}. Below we apply our theorem to solutions from \citet{koning2025continuous}, who provides optimal functions for problem~\eqref{eq:simple_opti} under several choices of $U(\cdot)$.
It will be convenient to introduce the notation,
$$
s^{\mathrm{ODP}}(x) = \frac{\sum_{j \in \mathcal{K}} q_j^{\circ}(x)}{\sum_{ j \in \mathcal{K}} p_j^{\circ}(x)}
$$
for the optimal discovery compound e-value function in~\eqref{eq:optimal_evalue}.

\begin{itemize}[leftmargin=*, wide, labelwidth=!, labelindent=0pt]
\item \emph{Generalized means.} Consider $U(s) = U_h(s) = (s^h - 1)/h$ for $h \in (0,1)$.\footnote{Note that as $h \to 0$, $U_h(s) \to \log(s)$, i.e., we recover the logarithmic utility function.} Moreover suppose that $ \E^{X' \sim \overline{\p}^{\circ}}[(s^{\mathrm{ODP}}(x))^{1/(1-h)}]  \in (0, \infty)$. Then a solution to~\eqref{eq:optim_compound_general} is given by:
$$
s^*(x) = s^*_h(x) = \frac{\left(s^{\mathrm{ODP}}(x)\right)^{1/(1-h)}}{\E^{X'\sim \overline{\p}^{\circ}}\left[\left(s^{\mathrm{ODP}}(x)\right)^{1/(1-h)}\right]}.
$$
\item \emph{Generalized means with clipping.} Fix $M \in (0,\infty)$ and $h \in (0,1)$ and define the utility $U(s)=U_{h,M}(s) = U_h(s \wedge M)$ with $U_h(\cdot)$ defined in the above bullet point. This utility function clips the argument at $M$, so that values exceeding $M$ provide no additional utility beyond what $M$ provides.  Then, provided that $\p_k( s(X_k) > 0) >0$ for at least one $k \in \mathcal{K}$,  a solution to~\eqref{eq:optim_compound_general} is given by:
$$
s^*(x) = s^*_{h,M}(x) = \frac{\left(s^{\mathrm{ODP}}(x)\right)^{1/(1-h)}\wedge M}{\E^{X'\sim \overline{\p}^{\circ}}\left[\left(s^{\mathrm{ODP}}(x)\right)^{1/(1-h)}\wedge M\right]}.
$$
\end{itemize}

\section{Proofs for Section~\ref{sec:further_defi}}

\subsection{Proof of Theorem~\ref{prop:compound-pe1}}

\begin{proof}
Fix an atomless $\p $ and it suffices to verify the statement for this $\p$. Let $K_0 = |\mathcal{N}(\p)|$ and take 
$L \sim \mathrm{Unif}(\mathcal{N}(\p))$ and 
$U\sim \mathrm{Unif}[0,1]$, 
mutually independent and independent of everything else. Then $$\tilde P:= P_{L} \id_{\{U\le K_0/K\}} +  \id_{\{ U>K_0/K\} }$$ is a bona fide p-value because for $t\in (0,1)$,
$$ \p( \tilde P  \leq t) =  \p(P_{L} \leq t) \p( U\leq K_0/K)  = \frac{K_0}{K} \frac{1}{K_0}\sum_{k: \p \in \mathcal{P}_k} \p(P_k \leq  t) \leq t.$$
Therefore, 
$$ \frac{1}{K} \sum_{k: \p \in \mathcal{P}_k }\E^{\p}\left[   h( P_k) \right] = \frac{K_0}{K} 
 \E^{\p}\left[   h( P_{L} ) \right] 
 \le   \E^{\p}\left[   h( P_L) \id_{\{U\le K_0/K\}}+ h(1)  \id_{\{ U>K_0/K\} }  \right]  =  \E^{\p}\left[ h(\tilde P) \right] \leq 1,$$
where the last inequality holds since $h$ is a p-to-e calibrator. 
\end{proof}

\section{Proofs for Section~\ref{sec:app-asym-com}}

\subsection{Proof of Proposition~\ref{prop:alter-p}}

\begin{proof}
When $\p$ is atomless, for any random variable $X$, there exists 
  an event $A$ with $\p(A)=1-\delta_\p$ such that $\{X < t \} \subseteq  A\subseteq  \{X\le  t \}$ for some $t \in \R $; see Lemma A.3 of \cite{wang2021axiomatic}.
  This fact will be used in the proof below. 
  If $\delta_\p=1$ there is nothing to prove, so we assume $\delta_\p<1$ below.

(P2) $\Rightarrow$ (P1): It suffices to notice  $$\p(P\le t) \le \p(P\le t,A)+\p(A^c) \le (1+\epsilon_\p) t +\delta_\p$$ for $\p\in \cP$ and $t\in [0,1]$. 

(P1) $\Rightarrow$ (P2) in case each $\p$ is atomless:  
Let $A$ be an event with $\p(A)=1-\delta_\p$ such that $\{P > s \} \subseteq  A\subseteq  \{P\ge  s \}$ for some $s \in [0,1]$.   
Note that $\p(P\le t,A^c)=\min\{\p(P\le t),\p(A^c)\}$.
Therefore, 
$$\p(P\le t, A) = 
( \p(P\le t) - \p(A^c))_+ \le ((1+\varepsilon_\p )t +\delta_\p-\delta_\p)_+ = (1+\varepsilon_\p)t$$
for all $t\in (0,1)$,
as desired.

(E2) $\Rightarrow$ (E1): It suffices to notice  $$\E^\p[E\wedge t] =
\E^\p[(E\wedge t) \id_A] + 
\E^\p[(E\wedge t) \id_{A^c}]
\le \E^\p[E \id_A] + 
\E^\p[t \id_{A^c}]
\le 
 (1+\varepsilon_\p ) +\delta_\p t $$ for $\p\in \cP$ and $t\in \R_+$. 

(E1) $\Rightarrow$ (E2) in case each $\p$ is atomless:  
Let $A$ be an event with $\p(A)=1-\delta_\p$ such that $\{E < t \} \subseteq  A\subseteq  \{E\le  t \}$ for some $t \in \R_+$. 
Note that $E\ge t$ on $A^c$.
We have, 
$$
\E^\p[E \id_A]   
= \E^\p[(E\wedge t)\id_A]
=  \E^\p[ E\wedge t ]
-\E^\p[ (E\wedge t) \id_{A^c}]
\leq 1+\epsilon_\p + t \delta_\p
-\E^\p[ t\id_{A^c}]
=  1+\epsilon_\p ,
$$ 
as desired. 
\end{proof}

\subsection{Proof of Theorem~\ref{th:calibration}}

\begin{proof}[Proof.]
Fix any $\p\in \cP$.   If $\delta_\p=1$ there is nothing to prove, so we assume $\delta_\p<1$ below. 
For the first statement, 
let $P$ be an $(\varepsilon, \delta)$-approximate  p-variable   for $\p$ and $f$ be a (p-to-e) calibrator and $b=(1+\epsilon_\p)/(1-\delta_\p) \ge 1$. 
Let $F$ be a cdf on $[0,1/b]$ given by
$F(s)=(1+\epsilon_\p)s + \delta_\p$ for $s\in [0,1/b]$.
We have $\p(P\le s   ) \le  F(s) $ for $s \in [0,1]$ and $F(0)=\delta_\p$.
Hence, noting that  the calibrator  $f$ is decreasing, we have, for all $t\in \R_+$, 
\begin{align*}
    \E^\p[f(P)\wedge t ] & 
    \le  \int_0^{1/b}( f(s)\wedge t) \d F(s)     
    \\ &\le  (1+\epsilon_\p) \int_0^{1/b}   f(s)\d s + t F(0) 
\\&\le  (1+\epsilon_\p) \int_0^1  f(s)\d s  +  \delta_\p t \le  1+\epsilon_\p + \delta_\p t. 
\end{align*}
For the second statement, define $P = (1/E)\wedge 1$. 
We have, for  $t\in (0,1)$,  
\begin{align*}
\p(P \leq t) 
&=\p(E\ge 1/t) 
\\ &=  \p(E\wedge(1/t) \ge 1/t) 
\\ &\le  t \E^\p[E\wedge(1/t) ] \leq t \left( 1+\epsilon_\p   +    \frac{\delta_\p}{t} \right)= (1+\epsilon_\p)  t +  \delta_\p,
\end{align*} 
as desired.
\end{proof}

\begin{remark}
Following a similar proof, the statements in Theorem~\ref{th:calibration} hold true also if we use the alternative formulations (P2) and (E2) in Proposition~\ref{prop:alter-p} for both approximate p-variables and approximate e-variables.
 \end{remark}

\subsection{Proof of Proposition~\ref{prop:varepsilon_to_delta}}

\begin{proof}
Let $P$ be an ($\varepsilon, \delta$)-approximate p-variable. Fix $\p \in \cP$. Then, for any $t \in (0, (1-\delta_{\p})/(1+\varepsilon_{\p})),$ we have that:

$$ \p(P \leq t) \leq (1+\varepsilon_{\p})t + \delta \leq t + \varepsilon_{\p}\frac{1-\delta_{\p}}{1+\varepsilon_{\p}} + \delta_{\p} = t + \delta'_{\p}.$$
Meanwhile the inequality $\p(P \leq t) \leq t + \delta'_{\p}$ is clearly true for $t \geq (1-\delta_{\p})/(1+\varepsilon_{\p})$ (since the right hand side is $\geq 1$). Thus $P$ is an $(0, \delta')$-approximate p-variable for $\cP$.

The argument for the approximate e-variable is analogous (noting that only $t\ge 1$ needs to be checked in the condition for approximate e-variables) and omitted. 
\end{proof}

\subsection{Proof of Proposition~\ref{prop:asym-e}}

\begin{proof}
We prove this result for the non-uniform case. The other argument is analogous. First suppose that $(E^{(n)})_{n \in \N}$ is a sequence of strongly   asymptotic e-variables for $\mathcal{P}$. Fix $\p \in \cP$. Then:
$$ \E^\p[E^{(n)}] \leq 1+ \varepsilon_n(\p) \to 1 \text{ as } n \to \infty,$$
where the first inequality holds for large enough $n$.
For the other direction, note that any nonnegative extended random variable $E^{(n)}$ is an $(\varepsilon_n,0)$-approximate e-variable for $\mathcal{P}$ with the choice $\varepsilon_n(\p) =  ( \E^\p[E^{(n)}]-1)_+$ as long as $\E^\p[E^{(n)}]$ is finite.
If  $\limsup_{n \to \infty} \E^\p[E^{(n)}]\le 1$, this implies that $\E^\p[E^{(n)}] < \infty$ for $n$ large enough and the above $\varepsilon_n(\p)$ satisfies $\lim_{n \to \infty} \varepsilon_n(\p) = 0$, that is, $(E^{(n)})_{n \in \N}$ is a sequence of strongly   asymptotic e-variables.
\end{proof}

\subsection{Proof of Proposition~\ref{proposition:uniform_integrability_strong_e}}

\begin{proof}
Fix $\p\in \cP$. Note that for $t\in \R_+$,
$$
\E^{\p}[E^{(n)}] \leq \E^{\p}[E^{(n)}\wedge t] + \E^{\p}[E^{(n)}\id_{\{E^{(n)} \geq t\}}].
$$
By uniform integrability, for any $\eta>0$ there exists $T>0$ such that 
$$\sup_{n \in \mathbb N}\E^{\p}\left[E^{(n)}\id_{\{E^{(n)} \geq T\}}\right]  \leq \eta.$$
Thus,
$$
\E^{\p}[E^{(n)}] \leq 1 + \varepsilon_n(\p) + \delta_n(\p) T + \eta,
$$
and hence,
$$ \limsup_{n \to \infty} \E^{\p}[E^{(n)}]  \leq 1+\eta.$$
Since $\eta>0$ is arbitrary, we conclude by Proposition~\ref{prop:asym-e}.
\end{proof}

\subsection{Proof of Proposition~\ref{prop:converge-dist}}

\begin{proof}
Fix $\p\in \cP$. 
Suppose that $(P^{(n)})_{n\in\N}$ converges in distribution to a p-variable $P$. Note that 
$ \limsup_{n\to\infty} \p(P_n\le t) \le \p(P\le t) \le t$ for each $t\in [0,1]$.
This implies 
$ \lim_{n\to\infty} \p(P_n\le t)\vee t =t$  for each $t\in [0,1]$. 
Since point-wise convergence of increasing functions to a continuous function on a compact interval implies uniform convergence,  we have 
$$\lim_{n\to\infty} \sup_{t\in [0,1]}(\p(P_n\le t)\vee t-t) = 0\implies \lim_{n\to\infty} \sup_{t\in [0,1]}(\p(P_n\le t)-t)_+ = 0. $$
This shows that $(P^{(n)})_{n\in\N}$ 
 is a sequence of asymptotic p-variables.

Next we show the statement on asymptotic e-variables. For $\delta \in [0,1]$,
and a nonnegative random variable $E$, 
consider the function
$$f_\delta(E)=\sup_{t\in \R_+} (\E^\p [E\wedge t]  -\delta  t) -1. $$ 
By definition, if $\epsilon :=
 f _\delta (E)
\vee 0 $ is finite, then $E$ is an $(\epsilon ,\delta )$-approximate e-variable for $\{\p\}$. 
Let $E$ be the limit of the nonnegative sequence $(E^{(n)})_{n\in\N}$ in distribution, and take $m\in \R_+$ be such that $\p(E\ge m)<\delta$. 
Since $E^{(n)}\to E$ in distribution, there exists $N\in \N$ such that $\p(E^{(n)}\ge m)<\delta$ for $n>N$ by the portmanteau theorem.
It follows that, for all $t\in \R_+$, %$\E^\p[E^{(n)}\wedge t]-\delta t \le \E^\p[E^{(n)}\wedge m]$ for $t\le m$
%and  
$$\E^\p[E^{(n)}\wedge t] - \E^\p[E^{(n)}\wedge m] - \delta t \le  
\p(E^{(n)} \ge m) (t-m)_+ - \delta t 
\le \delta (t-m)_+-\delta t
\le 0,
$$ 
and hence  $f_\delta(E^{(n)})\le \E^\p[E^{(n)}\wedge m]-1\to \E^\p[E\wedge m]-1\le 0$ as $n\to\infty$. 

This implies, in particular, that for any $\delta>0$,  
there exists $N'\in \N$ such that for $n>N'$, $E^{(n)}$ is an $(\delta,\delta)$-approximate e-variable. 
Since $\p\in \cP$ is arbitrary, this shows that $(E^{(n)})_{n\in\N}$ is a sequence of
 asymptotic e-variables for $\cP$.  
\end{proof}

\subsection{Proof of Example~\ref{example:weak_convergence_to_asymptotic}}

\begin{proof}
By the continuous mapping theorem, $P^{(n)}$ converges weakly to the uniform distribution, and hence item (i) directly follows from Proposition~\ref{prop:converge-dist}. 
 
For (ii) we may argue as follows. 
If  $h$ would have been continuous, then the result would follow 
from the continuous mapping theorem and  Proposition~\ref{prop:converge-dist}.
For upper semi-continuity, 
note that $\limsup_{n\to\infty} \p(E^{(n)}\ge  x) \le \p(h(Z)\ge x)$ for all $x\in \R$ by convergence in distribution. This gives 
$\limsup_{n\to \infty} f_\delta(E^{(n)}) \le f_\delta(h(Z))\le 0 $ using the notation and argument in the proof of Proposition~\ref{prop:converge-dist}. 
The last statement follows from 
Proposition~\ref{proposition:uniform_integrability_strong_e}.
When $h$ is increasing or decreasing, replacing it with its upper semi-continuous version, which does not change $\E[h(Z)]$, yields the desired statements. 
\end{proof}

\subsection{Proof of Proposition~\ref{proposition:approximate_evalues_to_approximate_compound}}

\begin{proof}
We carry out the argument only for compound e-variables (the argument for p-variables is analogous).
Take any $\p \in \cP$ and let $A_k$ be the event provided by Proposition~\ref{prop:alter-p} for the $k$-th ($\varepsilon_k, \delta_k)$-approximate e-variable. Then, letting $A := \bigcap_{k: \p \in \cP_k} A_k$, it follows that,
$$\p(A^c) = \p\left(\bigcup_{k:\p \in \cP_k} A_k^c\right)\leq \sum_{k:\p \in \cP_k} \delta_k(\p) = \delta(\p).$$ 
We get:
$$ \sum_{k: \p \in \cP_k}  \E^{\p}[E_k\id_{A}] \leq \sum_{k: \p \in \cP_k}  \E^{\p}[E_k\id_{A_k}] \leq \sum_{k: \p \in \cP_k} (1+\varepsilon_k(\p)) \leq K(1 + \varepsilon(\p)).$$
This proves the desired statement for compound e-variables.
\end{proof}

\subsection{Proof of Proposition~\ref{prop:equiv_approx_compound_e}}
\begin{proof}
Let us explicitly define the two statements under consideration:
\begin{enumerate}
\item[(E1)]  It holds that
$$
\E^\p \left[\left(\sum_{k: \p \in \cP_k}E_k \right)\wedge t \right] \le  K(1+\varepsilon_\p)  +\delta_\p t  \qquad \mbox{for all $t\in \R_+$ and all $\p\in \mathcal P$};
$$
\item[(E2)]  $E_1,\dots,E_K$ are $(\varepsilon, \delta)$-approximate compound e-variables for $(\cP_1,\dots,\cP_K)$ under $\cP$.
\end{enumerate}
We seek to show the equivalence of (E1) and (E2).

We start with (E1). By direct manipulation we see that (E1) is equivalent to the following statement:
$$
\E^\p \left[\left(\frac{1}{K}\sum_{k: \p \in \cP_k}E_k \right)\wedge t \right] \le  1+\varepsilon_\p  +\delta_\p t  \qquad \mbox{for all $t\in \R_+$ and all $\p\in \mathcal P$}.
$$
Now fix $\p \in \cP$. According to the above statement,  
$$E:=\frac{1}{K}\sum_{k: \p \in \cP_k}E_k \qquad \mbox{is an $(\varepsilon,\delta)$-approximate e-variable for $\{\p\}$}.$$
By Proposition~\ref{prop:alter-p} there exists an event $A$ with $\p(A) \geq 1-\delta$ such that $\E^{\p}[E\id_{A}] \leq 1+\varepsilon_\p$, that is,
$$ \sum_{k: \p \in \cP_k}  \E^{\p}[E_k\id_{A}] \leq K(1+\varepsilon_\p).
$$
Since $\p \in \cP$ was arbitrary, the above establishes the implication (E1) $\Rightarrow$ (E2). 

For the other direction,
fix an arbitrary $\p\in \cP$. Suppose (E2) holds, that is, $\sum_{k: \p \in \cP_k}  \E^{\p}[E_k\id_{A}] \leq K(1+\varepsilon_\p)$ for some event $A$ with $\p(A) \geq 1-\delta$. Define $E := (1/K) \sum_{k: \p \in \cP_k}E_k$, and  
by applying the direction (E2)$\Rightarrow $(E1) in Proposition~\ref{prop:alter-p}  with the hypothesis $\{\p\}$, we obtain the condition in (E1), completing the proof. 
\end{proof}

\subsection{Proof of Proposition~\ref{th:calibration_approximate_compound}}

\begin{proof}
For the first statement, 
let $P_1,\dotsc,P_K$ be $(\varepsilon, \delta)$-approximate compound p-variables. Fix any $\p\in \cP$ and let $h$ be a (p-to-e) calibrator. Let $A$ be an event with $\p(A) \ge  1-\delta_\p$ such that
$\sum_{k: \p \in \cP_k} \p(P\le t,A )\le K(1+\epsilon_\p) t$ for $t \in (0,1)$, 
and let $b=(1+\epsilon_\p)/\p(A)-1\ge 0.$
Let $\p_A(\cdot) = \p(\cdot \mid A)$. We have
$$\frac{1}{K}\sum_{k: \p \in \cP_k} \p_A(P_k\le t) \le (1+b)t \mbox{~~for $t \in (0,1)$}.$$
Define $U$,$L$, $P_L$, and $\tilde{P}$ as in the proof of Theorem~\ref{prop:compound-pe1}. Then,  $\tilde{P}$ satisfies,
$$
\p_A( \tilde P  \leq t) =  \p_A(P_{L} \leq t) \p_A( U\leq K_0/K)  = \frac{K_0}{K} \frac{1}{K_0}\sum_{k: \p \in \mathcal{P}_k} \p_A(P_k \leq  t) \leq (1+b)t.
$$
This means that $\tilde P$ is a $(b, 0)$-approximate p-variable for $\{\p_A\}$. By Theorem~\ref{th:calibration}, $h(\tilde{P})$ is a $(b,0)$-approximate e-variable for $\{\p_A\}$. Thus,
$$\frac{1}{K} \sum_{k: \p \in \mathcal{P}_k }\E^{\p_A}\left[   h( P_k) \right] = \frac{K_0}{K} 
 \E^{\p_A}\left[   h( P_{L} ) \right] 
 \le   \E^{\p_A}\left[ h(\tilde P) \right] \leq 1 + b =\frac{1+\epsilon_\p}{\p(A)}.$$
It follows that $\sum_{k: \p \in \mathcal{P}_k }\E^{\p}[h( P_k) \id_A] \leq K(1+\epsilon_\p)$, and so, $h(P_1),\dots,h(P_K)$ are $(\varepsilon,\delta)$-approximate compound e-variables.

The second claim follows by Markov's inequality.
Fix $\p \in \cP$ and let $A$ be the event given by the definition of approximate compound p-variables. Let $P_k = (1/E_k)\wedge 1$. Then, for any $t\in (0,1)$, 
\begin{align*}
 \sum_{k: \p \in \cP_k} \p(P_k \leq t,A) 
 &=  \sum_{k: \p \in \cP_k}  \p(E_k\ge 1/t,A) 
 \\ &=  \sum_{k: \p \in \cP_k}  \p(E_k\id_{A}  \ge 1/t) 
 \\ &\le   \sum_{k: \p \in \cP_k}  \E^\p[E_k\id_{A}] t \leq  K(1+\epsilon_\p)  t,
\end{align*} 
as desired.
\end{proof}

\section{Proofs for Section~\ref{sec:multiple_testing}}
\label{sec:ebh_admissible_proof}
\subsection{Proof of Theorem~\ref{theo:universality_eBH}}

\begin{proof}[Proof (continued).]
To show the last statement, consider an FDR procedure $\mathcal D$ at level $\alpha$ and take the compound e-values $E_1,\dots,E_K$ from  \eqref{eq:universal-e}. Define 
    $$
K^*:= \sup_{\p \in \mathcal{P}} \sum_{k: \p \in \mathcal{P}_k} \mathbb E^{\p}[E_k].
$$
If $K^*$ is equal to $K$, then there is nothing to show as $E_1,\dots,E_K$ are tight compound e-values. Otherwise, $K^*<K$. The case $K^*=0$ means one never rejects any hypotheses, for which choosing $E_1'=\dots=E_K'=1$ would suffice as tight compound e-values that produce $\mathcal D$ via e-BH.
For $K^*>0$, we let 
$$
E_k'=\frac{K}{K^*}E_k \mbox{~~~for $k\in \mathcal K$},
$$
and apply the e-BH procedure to $E'_1,\dots,E'_K$ (which are tight compound e-values). 
This new procedure controls FDR at level $\alpha$ and  produces at least as many discoveries as $\mathcal D$.
Since $\mathcal D$ is admissible,
this procedure must coincide with $\mathcal D$, and hence $\mathcal D$ is  e-BH applied to the tight compound e-values  $E'_1,\dots,E'_K$.
\end{proof}

\section{Proofs for Section~\ref{sec:sequence_model}}

\subsection{Proof of Proposition~\ref{prop:compound_separable}}

\begin{proof}
First, by definition of simple and separable, there exists a function $s$ such that $E_k = s(X_k)$. The supremum in the definition of tightness must be attained when $\mathcal{N}(\p) = K$, and so:
$$\sum_{k \in \mathcal{K}} \mathbb E^{\p_k^\circ}[s(X_k)]=K.$$
Now define $\overline{p}^\circ$ as the $\dd \nu$ density of $\overline{\p}^\circ$, where the latter object is defined as in~Theorem~\ref{th:compound-optimal}. Then the above may be written as:
$$\int s(x) \overline{p}^\circ(x) \dd \nu(x) = 1.$$
Thus, if we define $q_k^\circ(x) = s(x)\overline{p}^\circ(x)$ for all $k \in \mathcal{K}$, we find that $q_k^\circ$ is indeed a $\dd \nu$-density and $s(x)$ may be represented as in~\eqref{eq:optimal_evalue}.
\end{proof}

\subsection{Proof of Proposition~\ref{prop:lui}}

\begin{proof}
It suffices to show that $E_k^{\mathrm{LUI}}$ satisfies property (E2) of Proposition~\ref{prop:alter-p}. Fix $\p \in \cP_k$ and the $\psi_k$ such that $\dd \p_k/ \dd \nu = p^\circ_{\psi_k}$.  Note that $\psi_k$ is a function of $p_k$, but we keep this implicit in our notation. Now let $A$ be the event that \smash{$\psi_k \in \widehat{\Psi}_k(\delta)$}. 
Since \smash{$\widehat{\Psi}_k(\delta)$} is a $(1-\delta)$-confidence set for $\psi_k$, we have that $\p[A] \geq 1-\delta$. Moreover:
$$ 
 \E^{\p}\left[ E_k^{\mathrm{LUI}}\id_{A}\right] \leq \E^{\p}\left[\frac{q^\circ(X_k)}{ p^\circ_{\psi_k}(X_k)}\id_{A} \right] 
 \leq \E^{\p}\left[\frac{q^\circ(X_k)}{ p^\circ_{\psi_k}(X_k)} \right] = 1.
$$
Thus, $E_k^{\mathrm{LUI}}$ is $(0,\delta)$-approximate e-variable for $\cP_k$.
\end{proof}

\subsection{Proof of Proposition~\ref{prop:cui}}

\begin{proof}
Fix $\p \in \cP$ and so also $G(\boldpsi)$. Let $A$ be the event that $G(\boldpsi) \in \widehat{\mathcal{G}}$ which has probability at least $1-\delta$. On the event $A$, it holds by construction for all $k \in \mathcal{K}$ that
$$
E_k^{\mathrm{CUI}} \leq \frac{q^\circ(X_k)}{  \int p^\circ_{\psi}(X_k) \dd G(\boldpsi)(\psi)} =  \frac{K q^{\circ}(X_k)}{\sum_{ j \in \mathcal{K}} p_{\psi_j}^{\circ}(X_k)}.
$$
The objects on the right hand side are the optimal simple separable compound e-variables of Theorem~\ref{th:compound-optimal}. We can then conclude the proof similarly to the proof of Proposition~\ref{prop:lui}.
\end{proof}

\section{Proofs for Section~\ref{sec:ttest}}

\subsection{Proof of Theorem~\ref{theorem:eb_asymptotic}}

\begin{proof}

Fix $\p \in \cP^{\mathrm{N}}$. Also take $B>0$. We will split our analysis according to whether $\{\Norm{X_k} \leq B\}$ or $\{\Norm{X_k} > B\}$. In particular, we define:
$$
\begin{aligned}
\mathrm{I}_K(x; B) &:=  |E(x; \widehat{G}, \widehat{H}\otimes \widehat{G}) -  E(x; G(\boldsigma^2), H^*\otimes G(\boldsigma^2))|\id_{\{\Norm{x} \leq B\}},\\ 
\mathrm{II}_K(x;B) &:= E(x; \widehat{G}, \widehat{H}\otimes \widehat{G})\id_{\{\Norm{x} > B\}},
\end{aligned}
$$
where the dependence on $K$ is via $\widehat{G}$, $\widehat{H}$ and $G(\boldsigma^2)$.
Then for $k \in \mathcal{K}$
\begin{equation}
\begin{aligned}
E_k = E(X_k; \widehat{G}, \widehat{H}\otimes \widehat{G})  
 \,\leq\,  \mathrm{I}_K(X_k; B)\, +\, \mathrm{II}_K(X_k; B) \,+\, E(X_k; G(\boldsigma^2), H^*\otimes G(\boldsigma^2)).
\end{aligned}
\label{eq:proof_decomposition}
\end{equation}
Our strategy is as follows. We first show that for any fixed $B$
\begin{equation}
\max_{k \in \mathcal{K}: \p \in \cP_k} \mathbb E^{\p}[\mathrm{I}_K(X_k; B)] = o(1) \; \text{ as }\; K \to \infty. 
\label{eq:proof_I}
\end{equation}
Second, we show that
\begin{equation}
\sup_{K \in \mathbb N}\max_{k \in \mathcal{K}: \p \in \cP_k} \mathbb E^{\p}[\mathrm{II}_K(X_k; B)] = o(1) \; \text{ as }\; B \to \infty.
\label{eq:proof_II}
\end{equation}
Third, as we already noted in the main text, 
$E(X_k; G(\boldsigma^2), H^*\otimes G(\boldsigma^2))$ for $k \in \mathcal{K}$ are compound e-variables, i.e.,
\begin{equation}
\frac{1}{K} \sum_{k: \p \in \cP_k} \mathbb E^{\p}[E(X_k; G(\boldsigma^2), H^*\otimes 
G(\boldsigma^2))] \leq 1. 
\label{eq:proof_compound}
\end{equation}
By~\eqref{eq:proof_decomposition},~\eqref{eq:proof_I},~\eqref{eq:proof_II} and~\eqref{eq:proof_compound} and by first taking $K \to \infty$ and then $B \to \infty$, we find that:
$$
\limsup_{K \to \infty} \frac{1}{K} \sum_{k: \p \in \cP_k} E^{\p}[E_k] \leq 1,
$$
that is, $E_1,\ldots,E_K$ are strongly asymptotic e-variables. To conclude our proof, we still have to show~\eqref{eq:proof_I} and~\eqref{eq:proof_II}.

\paragraph{Proof of~\eqref{eq:proof_II}.}

First let $\bar{v} > 0$ be such that $|\lambda|/\sigma \leq \bar{v}$ uniformly over all 
$\lambda \in [\ubar{\lambda}, \bar{\lambda}]$ and $\sigma \in [\ubar{\sigma},\bar{\sigma}]$. 
Also let $\mathbf{1} \in \R^n$ be the vector of ones. Then, for any $\lambda \in [\ubar{\lambda}, \bar{\lambda}]$ and $\sigma \in [\ubar{\sigma},\bar{\sigma}]$,
$$p_{\lambda, \sigma^2}(X_k) =  p_{\sigma^2}(X_k) \exp\left(-\frac{n\lambda^2}{2}\right)\exp\left( \frac{\lambda \mathbf{1}^\intercal X_k}{\sigma} \right) \leq 2 p_{0, \sigma^2}(X_k) \cosh\left( \bar{v}\mathbf{1}^\intercal X_k \right).
$$
Singe $\widehat{G}$ is supported on $[\ubar{\sigma}^2, \bar{\sigma}^2]$ and $\widehat{H}$ on $[\ubar{\lambda},\bar{\lambda}]$, 
the above implies that
$$
\begin{aligned}
E(X_k; \widehat{G}, \widehat{H}\otimes \widehat{G}) &= \frac{\int p_{\lambda, \sigma^2}(X_k)  \dd (\widehat{H}\otimes \widehat{G})(\lambda, \sigma^2)}{\int p_{\sigma^2}(X_k) \dd \widehat{G}(\sigma^2)} \\ 
&\leq 2\cosh\left( \bar{v}\mathbf{1}^\intercal X_k \right) \frac{\int p_{\sigma^2}(X_k)  \dd (\widehat{H}\otimes \widehat{G})(\lambda, \sigma^2)}{\int p_{\sigma^2}(X_k) \dd \widehat{G}(\sigma^2)}  \\ 
&=  2\cosh\left( \bar{v}\mathbf{1}^\intercal X_k \right).
\end{aligned}
$$
Thus,
$$
\begin{aligned}
\sup_{K \in \mathbb N}\max_{k \in \mathcal{K}: \p \in \cP_k} \mathbb E^{\p}[\mathrm{II}_K(X_k; B)] &\leq 2\sup_{k \geq 1: \p \in \cP_k}\mathbb E^{\p}\left[\cosh\left( \bar{v}\mathbf{1}^\intercal X_k \right)  \id_{\{\Norm{X_k} > B\}}\right] \\ 
& \leq  2\sup_{k \geq 1: \p \in \cP_k} \left\{\p( \id_{\{\Norm{X_k} > B\}})\right\}^{1/2}  \sup_{k \geq 1: \p \in \cP_k}  \left\{E^{\p}\left[\cosh^2\left( \bar{v}\mathbf{1}^\intercal X_k \right)\right]\right\}^{1/2}. 
\end{aligned}
$$
The first term converges to $0$ as $B \to \infty$ because the class of distributions $\{\mathrm{N}(0, \sigma^2 I)\,:\, \sigma^2 \in [\ubar{\sigma}^2, \bar{\sigma}^2]\}$ is tight. Meanwhile, the second term remains bounded by a direct calculation of the moment generating function of a Gaussian, and using again the fact that for all $k$, $\sigma_k^2 \in [\ubar{\sigma}^2, \bar{\sigma}^2]$.

\paragraph{Proof of~\eqref{eq:proof_I}.}
We start with an argument inspired by~\citet{greenshtein2022generalized}.
Write $G_K = G(\boldsigma^2) \in \mathcal{G}$, and note that the class of distributions $\mathcal{G}$ is tight. Take any subsequence $K_{\ell}$, then there exists a further subsequence $K_{\ell_{m}}$ and a distribution $G^*$ such that \smash{$G_{K_{\ell_m}} \cd G^*$}. By a standard subsequence argument, it suffices to show that  
$$\max_{k \in \{1,\ldots,K_{\ell_m}\}: \p \in \cP_k} \mathbb E^{\p}[\mathrm{I}_{K_{\ell_m}}(X_k; B)] = o(1) \; \text{ as }\; m \to \infty.$$ 
In what follows, to streamline notation, we assume (without loss of generality, via the subsequence argument above) that there exists $G^*$ such that
\begin{equation}
\label{eq:empirical_dbn_converges}
G_K \cd G^*\, \text{ as }\, K \to \infty.
\end{equation}
Then also $H^* \otimes G_K \cd H^* \otimes G^*$. We will first show that
\begin{equation}
\sup_{x:\Norm{x} \leq B} |E(x; G^*, H^* \otimes G^*) -  E(x; G(\boldsigma^2), H^*\otimes G(\boldsigma^2))| = o(1) \; \text{ as } \; K \to \infty.
\label{eq:evalue_g_star_g_sigma}
\end{equation}
To this end, introduce the notation
$$
N(x; Q) := \int p_{\lambda, \sigma^2}(x)  \dd Q(\lambda, \sigma^2),\;\;\; D(x;G):= \int p_{\sigma^2}(x) \dd G(\sigma^2),
$$
so that $E(x; G,Q) = N(x; Q)/D(x;G)$. Note that the collection of functions parameterized by $x$,
\begin{equation}
\left\{ [\ubar{\lambda},\bar{\lambda}]\times [\ubar{\sigma}^2, \bar{\sigma}^2] \to \mathbb R_+, \;\;\; \ (\lambda, \sigma^2) \mapsto p_{\lambda,\sigma^2}(x)\;\;\;:\;\;\; \Norm{x} \leq B \right\},
\label{eq:fct_class_1}
\end{equation}
is uniformly bounded and equicontinuous in $(\lambda, \sigma^2)$.  Fix $\varepsilon >0$. Then, by Arzelà–Ascoli, for any $\varepsilon >0$, there exists a finite set of continuous bounded functions $\mathcal{M}_N(\varepsilon)$ that provide a supremum-norm cover of the functions in~\eqref{eq:fct_class_1}. That is, for any $x$ such that $\Norm{x} \leq B$, there exists $m_N \in \mathcal{M}_N(\varepsilon)$ satisfying
$$
\sup_{ (\lambda,\sigma^2) \in [\ubar{\lambda},\bar{\lambda}]\times [\ubar{\sigma}^2, \bar{\sigma}^2]} |p_{\lambda, \sigma^2}(x) - m_N(\lambda,\sigma^2)| \leq \varepsilon.
$$
This also implies that for $Q,Q'$ supported on $[\ubar{\lambda},\bar{\lambda}]\times [\ubar{\sigma}^2, \bar{\sigma}^2]$, we have that,
$$
\left|  N(x; Q) - N(x;Q') \right| \;\leq \; 2\varepsilon \,+\, \left|  \int m_N(\lambda,\sigma^2)\dd Q(\lambda, \sigma^2) -  \int m_N(\lambda,\sigma^2)\dd Q'(\lambda, \sigma^2)    \right|.
$$
Repeating this argument for all $x$ in the above ball, we find that
$$
\sup_{ x: \Norm{x} \leq B}\left|  N(x; Q) - N(x;Q') \right|  \;\leq \; 2\varepsilon \,+\,  \max_{m_N \in \mathcal{M}_N(\varepsilon)} \left|  \int m_N(\lambda,\sigma^2)\dd Q(\lambda, \sigma^2) -  \int m_N(\lambda,\sigma^2)\dd Q'(\lambda, \sigma^2)    \right|.
$$
Arguing analogously for the collection of functions
\begin{equation}
\left\{  [\ubar{\sigma}^2, \bar{\sigma}^2] \to \mathbb R_+, \;\;\; \  \sigma^2 \mapsto p_{\sigma^2}(x)\;\;\;:\;\;\; \Norm{x} \leq B \right\},
\label{eq:fct_class_2}
\end{equation}
we find that for any $\varepsilon >0$, there exists a finite collection of functions $\mathcal{N}_D(\varepsilon)$ such that for any $G,G'$ supported on $[\ubar{\sigma}^2, \bar{\sigma}^2]$ it holds that,
$$
\sup_{ x: \Norm{x} \leq B}\left|  D(x; G) - D(x;G') \right|  \;\leq \; 2\varepsilon \,+\,  \max_{m_D \in \mathcal{M}_D(\varepsilon)} \left|  \int m_D(\sigma^2)\dd G(\sigma^2) -  \int m_D(\sigma^2)\dd G'(\sigma^2)    \right|.
$$
Also let us define,
$$D(B):= \inf \left\{  p_{\sigma^2}(x)\,:\, \Norm{x} \leq B,\, \sigma^2 \in [\ubar{\sigma}^2, \bar{\sigma}^2] \right\},$$
where we make only the dependence on $B$ explicit. By a continuity and compactness argument, $D(B) \in (0,\infty)$. Then, given $G,G',Q,Q'$ we have that:
$$
\begin{aligned}
&\sup_{ x: \Norm{x} \leq B}\left|E(x; G,Q) -  E(x; G',Q')  \right|   \\ 
&\;\;\; = \sup_{ x: \Norm{x} \leq B}\left| \frac{N(x;Q)}{D(x;G)} -  \frac{N(x;Q')}{D(x;G')}\right| \\  
&\;\;\; \leq \frac{1}{D(B)} \sup_{ x: \Norm{x} \leq B}| N(x;Q) - N(x;Q')| \,+\,  \frac{1}{D(B)^2} \sup_{ x: \Norm{x} \leq B}|D(x;G)-D(x;G')| \\ 
&\;\;\; \leq \frac{2\varepsilon}{D(B)} \,+\,  \frac{1}{D(B)}\max_{m_N \in \mathcal{M}_N(\varepsilon)} \left|  \int m_N(\lambda,\sigma^2)\dd Q(\lambda, \sigma^2) -  \int m_N(\lambda,\sigma^2)\dd Q'(\lambda, \sigma^2)\right| \\
&\;\;\;\;\; + \; \frac{2\varepsilon}{D(B)^2} \,+\,  \frac{1}{D(B)^2} \max_{m_D \in \mathcal{M}_D(\varepsilon)} \left|  \int m_D(\sigma^2)\dd G(\sigma^2) -  \int m_D(\sigma^2)\dd G'(\sigma^2)    \right|
\end{aligned}
$$
Fix any $m_D \in \mathcal{M}_D(\varepsilon)$. Recalling that $m_D$ is bounded and continuous (on its support), and that $G_K \cd G^*$, we find that
$$\int m_D(\sigma^2)\dd G_K(\sigma^2) \to  \int m_D(\sigma^2)\dd G^*(\sigma^2) \;\text{ as }\; K \to \infty.$$
We can argue analogously for any $m_N \in \mathcal{M}_N(\varepsilon)$ by recalling that 
$H^* \otimes G_K \cd H^* \otimes G^*$. Thus, by taking $K \to \infty$ and then $\varepsilon \to 0$, we have shown~\eqref{eq:evalue_g_star_g_sigma}.
In Lemma~\ref{lemma:npmle_cd} below, we will show that~\eqref{eq:empirical_dbn_converges} implies that
$$
\widehat{G} \cd G^* \text{ almost surely as } K \to \infty,
$$
and so also $\widehat{Q} = \widehat{H} \otimes \widehat{H} \cd H^* \otimes G^*$ almost surely. 
The same arguments as above then show that,
$$
\sup_{ x: \Norm{x} \leq B}\left|E(x;  \widehat{G}, \widehat{H}\otimes \widehat{G}) -  E(x; G^*,H^*\otimes G^*)  \right|  \to 0 \; \text{ almost surely as }\; K \to \infty,
$$
so that by the triangle inequality and~\eqref{eq:evalue_g_star_g_sigma},
$$
\sup_{ x: \Norm{x} \leq B}\left|E(x;  \widehat{G}, \widehat{H}\otimes \widehat{G}) -  E(x; G(\boldsigma^2),H^*\otimes G(\boldsigma^2))  \right|  \to 0 \; \text{ almost surely as }\; K \to \infty.
$$
We can check that the left-hand side is uniformly bounded. Thus, by dominated converge,
$$
\mathbb E^{\mathbb P}\left[ \sup_{ x: \Norm{x} \leq B}\left|E(x;  \widehat{G}, \widehat{H}\otimes \widehat{G}) -  E(x; G(\boldsigma^2),H^*\otimes G(\boldsigma^2))\right|\right] = o(1) \;\text{ as } \; K \to \infty.
$$
The expression on the left-hand side above provides an upper bound on $\mathrm{I}_K(x; B)$ for any $x$ with $\Norm{x} \leq B$, and thus we have shown~\eqref{eq:proof_I}.

\end{proof}

\begin{lemma}[NPMLE convergence]
\label{lemma:npmle_cd}In the setting of this section, suppose that~\eqref{eq:empirical_dbn_converges} holds. Then:
$$
\widehat{G} \cd G^* \text{ almost surely as } K \to \infty.
$$
\end{lemma}

\begin{proof}
For two densities $f,h$ on $\mathbb R_+$, we define their squared Hellinger distance ($\Dhel^2(f,h)$) and total variation distance ($\TV(f,h)$) via,
$$
\Dhel^2(f,h) := \frac{1}{2}\int_0^{\infty} \left( \sqrt{f(u)}  - \sqrt{h(u)}\right)^2\,\dd u ,\;\;\quad \TV(f,h):=\frac{1}{2}\int_0^{\infty} \left| f(u)  - h(u)\right|\,\dd u.
$$
By~\eqref{eq:empirical_dbn_converges}, for any $u>0$:
$$ f_{G_K}(u) = \int p(u \mid \sigma^2) \dd G_K(\sigma^2) \to \int p(u \mid \sigma^2 ) \dd G^*(\sigma^2) =  f_{G^*}(u)\; \text{ as } k \to \infty.$$
Here we used the definition of weak convergence along with the fact that $\sigma^2 \mapsto p(u \mid \sigma^2)$ 
is continuous and bounded on $[\ubar{\sigma}^2, \bar{\sigma}^2]$ (this requires that $n \geq 3$). Hence by Scheffe's theorem~\citep[Theorem 16.12.]{billingsley1995probability}
applied to the Lebesgue measure on $(0, \infty)$, we get convergence in total variation distance ($\TV$),
$$\TV( f_{G_K},\, f_{G^*})  \to \, 0\; \text{ as }\; m \to \infty.$$
Meanwhile, \citet[Theorem 9$^*$]{ignatiadis2024empirical} along with the Borel-Cantelli lemma show that
\begin{equation}
\Dhel^2(f_{G_K}, f_{\widehat{G}}) \to 0\; \text{ almost surely as }\; K \to \infty\;  ,
\label{eq:hellinger_convergence}
\end{equation}
where we use the notation from the proof of Theorem~\ref{theorem:eb_asymptotic} and write $G_K = G(\boldsigma^2)$. In what follows we will assume that we are on the (almost sure) event wherein the convergence in~\eqref{eq:hellinger_convergence} holds (and all objects defined henceforth are to be interpreted as functions on the sample space of the underlying probability space).
We also make the dependence of $\widehat{G}$ on $K$ explicit by writing $\widehat{G}_K \equiv \widehat{G}$.

By construction it holds that $(\widehat{G}_K)_{K \geq 1}$ is tight. Next take any subsequence $(\widehat{G}_{K_{m}})_{m \geq 1}$ such that $\widehat{G}_{K_m} \cd \widetilde{G}$
as $m \to \infty$ for some probability measure $\widetilde{G}$. By~\citet[Chapter 25, corollary on page 337]{billingsley1995probability}, 
the weak convergence $\widehat{G}_K \cd G^*$ will follow if we can show that $\widetilde{G} = G^*$.

To see that this holds, first argue as above to show that for any $u>0$
$$ f_{\widehat{G}_{K_m}}(u) = \int p(u \mid \sigma^2) \dd \widehat{G}_{K_m}(\sigma^2) \to \int p(u \mid \sigma^2 ) \dd \widetilde{G}(\sigma^2) =  f_{\widetilde{G}}(u)\; \text{ as } k \to \infty,$$
and also that
$$\TV( f_{\widehat{G}_{K_m}},\, f_{\widetilde{G}})  \to \, 0\; \text{ as }\; m \to \infty.$$
On the other hand:
$$ \TV(f_{\widehat{G}_{K_m}}, \, f_{G_{K_m}} ) \stackrel{(*)}{\leq} \sqrt{2} \Dhel(f_{\widehat{G}_{K_m}}, \, f_{G_{K_m}})   \, \stackrel{(**)}{\to} \, 0\; \text{ as }\; m \to \infty.$$
$(**)$ follows because we are on the event wherein the convergence in~\eqref{eq:hellinger_convergence} holds, and $(*)$ is a standard fact that follows, 
e.g., by the Cauchy-Schwarz inequality. 

Using the triangle inequality, we thus have:
$$
\TV(f_{\widetilde{G}}, \, f_{G^*}) \leq  \TV(f_{\widehat{G}_{K_m}}, \, f_{G_{K_m}}) + \TV( f_{\widehat{G}_{K_m}},\, f_{\widetilde{G}}) + \TV( f_{G_{K_m}},\, f_{G^*}) \to 0\;\text{ as }\; m \to \infty,
$$
i.e., \smash{$\TV(f_{\widetilde{G}}, \; f_{G^*} ) =0$}, which implies 
that $f_{G^*}(u) = f_{\widetilde{G}}(u)$ for almost all $u>0$, and so by continuity, for all $u>0$.
Let $\Gamma^*$ and \smash{$\widetilde{\Gamma}$} be the push-forwards of $G^*$ and \smash{$\widetilde{G}$} under the map $\sigma^2 \mapsto 1/\sigma^2$.  
Then it follows by~\citet[Section 4, Example 1]{teicher1961identifiability} that \smash{$\Gamma^* = \widetilde{\Gamma}$}.
This in turn implies that \smash{$G^* = \widetilde{G}$}.

\end{proof}

\subsection{Proof of Theorem~\ref{prop:c9-asym}}

\begin{proof}[Proof.]
Take any $\p \in \cP$. By the central limit theorem, the law of large numbers and Slutsky's theorem, $\sqrt{n}\bar{X}^{(n)}/S^{(n)}$ converges in distribution to $\mathrm{N}(0,1)$ as $n \to \infty$. 
We may verify that
$$ \mathbb E^{\mathbb P} \left[\exp\left(\lambda Z - \frac{\lambda^2}{2}\right)\right] =1,$$
for $Z \sim \mathrm{N}(0,1)$. 
By Example~\ref{example:weak_convergence_to_asymptotic} it follows that $E^{(n)}$ are asymptotic e-variables. To upgrade this result to strongly asymptotic, we can
apply Theorem 2.5 of~\citet{gine1997when} to show that the sequence $(E^{(n)})_{n \in \N}$ is uniformly integrable and then conclude the proof by Proposition~\ref{proposition:uniform_integrability_strong_e}.

Next take an alternative $\q$ with $\E^{\mathbb Q}[X^i]> 0$, $\var^{\mathbb Q}(X^i) < \infty$, and suppose that $\lambda >0 $. Observe the following. By the strong law of large numbers
$$ 
\bar{X}^{(n)} \to \E^{\mathbb Q}[X^i]>0\;\; \text{ almost surely as } \;\; n \to \infty,
$$
and
$$
S^{(n)} \to \sqrt{  \var^{\mathbb Q}({X^i}) + \E^{\mathbb Q}[X^i]^2 }>0  \;\; \text{ almost surely as } \;\; n \to \infty.
$$
The above imply that
$$
\lambda \frac{\sqrt{n} \bar{X}^{(n)}}{S^{(n)}} \to \infty\;\; \text{ almost surely as } \;\; n \to \infty,
$$
which also means that $E^{(n)}$ goes to $\infty$ with probability 1.
The argument for $\tilde{E}^{(n)}$ is analogous and omitted. The argument for $\hat{\sigma}^{(n)}$ follows from Example~\ref{example:weak_convergence_to_asymptotic}.
\end{proof}

\subsection{Proof of Proposition~\ref{proposition:simultaneous_ttest_variance_compound}}

\begin{proof}
Fix an arbitrary distribution $\p \in (\bigcup_k \cP_k) \bigcap \cP^{\mathrm{B}}$ and write $\sigma_k^2 = \sigma_k^2(\p)$, omitting the dependence on $\p$.
For each $k$, we define
$$
\tilde{X}_{k}^i := X_{k}^i - \mu_{k}(\p),
$$
so that $\tilde{X}_{k}^i = X_k^i$ for all $k \in \mathcal{N}(\p)$. We also define $\tilde{S}_k^{(n)}$ analogously to $S_k^{(n)}$ in~\eqref{eq:one_sample_summary_mtp_rewrite} but with the $X_{k}^i$ replaced by $\tilde{X}_k^i$. Then,
note that 
\begin{equation}
\label{eq:asymp_comp_evalue_renormalized}
\frac{1}{K}\sum_{k: \p \in \cP_k} E_k^{(n)} =  \frac{\sum_{k:\p \in \cP_k } \big(S_k^{(n)}\big)^2}{\sum_{k \in \mathcal{K}}\big(\hat{\sigma}^{(n)}_{k}\big)^2}     \leq \frac{\sum_{k \in \mathcal{K} } \big(\tilde{S}_k^{(n)}\big)^2}{\sum_{k \in \mathcal{K}}\big(\hat{\sigma}^{(n)}_{k}\big)^2}.
\end{equation}
We will argue that the right hand side of the above expression converges in probability (and thus also in distribution) to the constant $1$. Arguing as in Proposition~\ref{prop:converge-dist} it will then follow that $E_1^{(n)},\ldots, E_K^{(n)}$ are asymptotic compound e-variables. We claim that
$$
\frac{n\sum_{k \in \mathcal{K} } \big(\tilde{S}_k^{(n)}\big)^2}{ n\sum_{k \in \mathcal{K} } \sigma_k^2} \stackrel{\p}{\to} 1 \;\;\text{ as }\;\; m \to \infty.
$$
Writing $U^{(m)} = n\sum_{k \in \mathcal{K} } \big(\tilde{S}_k^{(n)}\big)^2 = \sum_{k \in \mathcal{K}} \sum_{i=1}^n (\tilde{X}_{k}^i)^2$, we see that $\E^{\p}[U^{(m)}] = n\sum_{k \in \mathcal{K} } \sigma_k^2$. Moreover, by our assumptions and Marcinkiewicz–Zygmund, there exists another constant $C'$ that depends only on $\delta$ such that:
$$
\E^{\p}[U^{(m)}] \geq cnK,\quad \E^{\p}[|U^{(m)} - \E^{\p}[U^{(m)}]|^{1+\delta}] \leq C' C n K.
$$
Thus, since $n K \to \infty$ and using Markov's inequality, it follows that $U^{(m)}/ \E^{\p}[U^{(m)}]$ converges in probability to $1$ as desired. We can analogously argue that:
$$
\frac{n\sum_{k \in \mathcal{K} } \big(\hat{\sigma}^{(n)}_{k}\big)^2}{ n\sum_{k \in \mathcal{K} } \sigma_k^2} \stackrel{\p}{\to} 1 \;\;\text{ as }\;\; m \to \infty.
$$
Thus implies that $E_1^{(n)},\ldots, E_K^{(n)}$ are asymptotic compound e-variables.

Next fix $\p \in (\bigcup_k \cP_k) \bigcap \cP^{\mathrm{N}}$. It suffices to prove uniform integrability of the right hand side in~\eqref{eq:asymp_comp_evalue_renormalized} for all $m$ such that $(n-1)K \geq 5$  (using an argument analogous to the one of Proposition~\ref{proposition:uniform_integrability_strong_e}). We prove uniform integrability by controlling the following $\eta$-th moment for $\eta \in (1, 1.25)$:
$$
\E^{\p}\left[ \left( \frac{\sum_{k \in \mathcal{K} } \big(\tilde{S}_k^{(n)}\big)^2}{\sum_{k \in \mathcal{K}}\big(\hat{\sigma}^{(n)}_{k}\big)^2} \right)^{\eta} \right] \leq \E^{\p}\left[ \left(\sum_{k \in \mathcal{K} } \big(\tilde{S}_k^{(n)}\big)^2\right)^{2\eta}\right]^{1/2} \E^{\p}\left[ \left( \frac{1}{\sum_{k \in \mathcal{K}}\big(\hat{\sigma}^{(n)}_{k}\big)^2} \right)^{2\eta} \right]^{1/2}.
$$
Notice that by our assumptions, $\sigma_k^2=\sigma_k^2(\p)$ is both uniformly lower bounded and upper bounded for all $k$. Ignoring multiplicative constants, it suffices to bound:
$$
\E^{\p}\left[ \left(\sum_{k \in \mathcal{K} } n\big(\tilde{S}_k^{(n)}/\sigma_k\big)^2\right)^{2\eta}\right]^{1/2} \E^{\p}\left[ \left( \frac{1}{\sum_{k \in \mathcal{K}}(n-1)\big(\hat{\sigma}^{(n)}_{k}/\sigma_k \big)^2} \right)^{2\eta}\right]^{1/2}. 
$$
This can be accomplished by noting that,
$$\sum_{k \in \mathcal{K}} n\big(\tilde{S}_k^{(n)}/\sigma_k\big)^2 \sim \chi^2_{nK},\;\;\;\;\; \sum_{k \in \mathcal{K}}(n-1)\big(\hat{\sigma}^{(n)}_{k}/\sigma_k \big)^2 \sim \chi^2_{(n-1)K} ,$$
where $\chi^2_d$ denotes the chi-square distribution with $d$ degrees of freedom. The remainder of the calculation amounts to direct evaluation of the $2\eta$-th moment (resp.~negative $2\eta$-th moment) of chi squared random variables.
\end{proof}

\end{document}